\documentclass{article}

\usepackage[preprint,nonatbib]{neurips_2022}

\usepackage[utf8]{inputenc} 
\usepackage[T1]{fontenc}    
\usepackage{hyperref}       
\usepackage{url}            
\usepackage{booktabs}       
\usepackage{amsfonts}       
\usepackage{nicefrac}       
\usepackage{microtype}      
\usepackage{xcolor}         

\usepackage{amsmath}
\usepackage{amsthm}
\usepackage{amsfonts}
\usepackage{amssymb}
\usepackage{mathtools} 
\usepackage{bbm} 
\usepackage{bm} 
\usepackage{graphicx}
\usepackage[ruled,vlined,linesnumbered]{algorithm2e}
\usepackage{thmtools}
\usepackage{thm-restate}

\DeclareMathOperator*{\argmax}{arg\,max}
\newtheorem{claim}{Claim}
\newtheorem{definition}{Definition}

\newtheorem{lemma}{Lemma}
\newtheorem{theorem}{Theorem}
\newtheorem{corollary}{Corollary}

\newcommand{\1}{\mathbbm{1}}

\def\ddefloop#1{\ifx\ddefloop#1\else\ddef{#1}\expandafter\ddefloop\fi}
\def\ddef#1{\expandafter\def\csname bb#1\endcsname{\ensuremath{\mathbb{#1}}}}
\ddefloop ABCDEFGHIJKLMNOPQRSTUVWXYZ\ddefloop

\def\ddef#1{\expandafter\def\csname c#1\endcsname{\ensuremath{\mathcal{#1}}}}
\ddefloop ABCDEFGHIJKLMNOPQRSTUVWXYZ\ddefloop

\def\ddef#1{\expandafter\def\csname v#1\endcsname{\ensuremath{\boldsymbol{#1}}}}
\ddefloop ABCDEFGHIJKLMNOPQRSTUVWXYZabcdefghijklmnopqrstuvwxyz\ddefloop

\def\ddef#1{\expandafter\def\csname v#1\endcsname{\ensuremath{\boldsymbol{\csname #1\endcsname}}}}
\ddefloop {alpha}{beta}{gamma}{delta}{epsilon}{varepsilon}{zeta}{eta}{theta}{vartheta}{iota}{kappa}{lambda}{mu}{nu}{xi}{pi}{varpi}{rho}{varrho}{sigma}{varsigma}{tau}{upsilon}{phi}{varphi}{chi}{psi}{omega}{Gamma}{Delta}{Theta}{Lambda}{Xi}{Pi}{Sigma}{varSigma}{Upsilon}{Phi}{Psi}{Omega}{ell}\ddefloop

\usepackage[backend=bibtex,style=numeric]{biblatex}
\newcommand{\sacomment}[1]{}
\newcommand{\sycomment}[1]{}

\newcommand{\newshipra}[1]{#1}
\newcommand{\synew}[1]{#1}
\newcommand{\syremoved}[1]{}

\newcommand{\OPT}{\text{OPT}}
\newcommand{\Regret}{\text{Regret}}

\newcommand{\Xvec}{\boldsymbol{X}}
\newcommand{\Fvec}{\boldsymbol{F}}
\newcommand{\Gvec}{\boldsymbol{G}}
\newcommand{\upperx}{\bar x}

\newcommand{\lowerx}{\underline{x}}

\title{Online Allocation with Strategic Agents and Unknown Distribution}
\title{Online Allocation and Learning in the Presence of Strategic Agents}

\author{%
  Steven Yin\\
  Department of Industrial Engineering and Operations Research\\
  Columbia University\\
  New York, NY 10027 \\
  \texttt{sy2737@columbia.edu} \\
  \And
  Shipra Agrawal \\
  Department of Industrial Engineering and Operations Research\\
  Columbia University\\
  New York, NY 10027 \\
  \texttt{sa3305@columbia.edu} \\
   \AND
   Assaf Zeevi \\
   Graduate School of Business\\
   Columbia University\\
   New York, NY 10027 \\
   \texttt{assaf@gsb.columbia.edu} \\
}

\begin{document}

\maketitle

\begin{abstract}
     We study the problem of allocating $T$ sequentially arriving items among $n$ homogeneous
  agents under the constraint that each agent must receive a pre-specified fraction of
  all items, with the objective of maximizing the agents' total valuation of
  items allocated to them. The agents' valuations for the item in each round are
  assumed to be i.i.d. but their distribution is a priori unknown to the central
  planner. Therefore, the central planner needs to implicitly learn these
  distributions from the observed values in order to pick a good allocation policy.
  However, an added challenge here is that the agents are strategic with incentives to
  misreport their valuations in order to receive better allocations. This sets 
  our work apart both from the online auction design settings which
  typically assume known valuation distributions and/or involve payments, and from the 
  online learning settings that do not consider strategic agents. To that end, our main
  contribution is an online learning based allocation mechanism
  that is approximately Bayesian incentive compatible, and when all agents are truthful, guarantees a sublinear regret for individual agents' utility compared to that under the optimal offline allocation policy.
\end{abstract}

\section{Introduction}
\syremoved{Growth in online marketplaces have driven a lot of recent academic interest in
online resource allocation. For example, online ad platforms need to
sequentially allocate impressions to one of many advertisers. Maximizing some
notion of global objective, such as social welfare or revenue, while maintaining
certain fairness or budget constraints, have become a key research problem at
the heart of many for-profit and non-profit organizations. }

A classic sequential resource allocation problem is to allocate $T$ sequentially arriving items to $n$ agents, where each agent must receive a predetermined fraction of the  items. The goal is to maximize social welfare, i.e., the agents' total valuation of the items allocated to them. This problem is non-trivial even when the agents' valuations are stochastic and i.i.d. with a known distribution, the main difficulty being that the allocations must be performed in real-time; specifically, an item must be allocated to an agent in the current round without knowledge of their future valuations. 

A more challenging (and quite useful) extension of the problem which has been the focus of recent literature considers the case where the distribution of the agents' valuations is apriori unknown to the planner. In such settings, algorithms based on online learning can be used to adaptively learn the valuation distribution from observed valuations in previous rounds, and improve the allocation policy over time (see \cite{agrawal2009dynamic, DevanurSivan2019, balseiro2020best,balseiro2021regularized} for some examples). 
However, these mechanisms implicitly assume that the agents report their valuations truthfully, so that the mechanism can directly learn from the reported valuations in order to maximize the social welfare.

Many practical resource allocations settings do not conform with the truthful reporting assumption. In particular, selfish and strategic agents may have an incentive to misreport their valuations if that can lead to individual utility gain (possibly at the expense of social welfare). 
Hence, an allocation policy that does not take such misreporting incentives into account can incur
significant loss in social welfare in presence of strategic agents. 
\syremoved{For example, consider a simple setting with two agents whose true valuations are
respectively $\{0, 1\}$ w.p. $1/2$, and $\{0,1-\epsilon\}$ w.p. $1/2$; and  each agent must be allocated exactly half of the $T$ items.  If both agents report their valuations truthfully,  then under the optimal (social welfare maximizing) allocation policy, the second agent will receive only $T/4$ (in expectation) items with value $1-\epsilon$  and the remaining $T/4$ items with value $0$. This provides an incentive for the second agent to single-handedly deviate from the truthful reporting strategy and instead report a very high valuation whenever her valuation is $1-\epsilon$. With this misreporting, the second agent would receive $T/2$ items of valuation $1-\epsilon$, and the social welfare will decrease by $T\epsilon/4$. }
\synew{For example, consider a simple setting with two agents whose true valuations are i.i.d. and uniformly distributed between 0 and 1. That is,
$$X_1, X_2 \stackrel{i.i.d.}{\sim} \text{Uniform[0,1]}$$
Each agent is pre-determined to receive an equal fraction of all the items. 
The optimal welfare maximizing allocation policy is to allocate the item to the agent with higher valuation in (almost) every round. 
This policy results in $T/3$ expected utility ($\bbE[X_1 | X_1>X_2]/2$) for each agent, and  a social welfare of $2T/3$. However, suppose that the first agent  chooses to misreport in the following way: the agent reports a high valuation of $1$ whenever her true valuation is in $[0.5,1]$
and a low valuation of $0$ whenever her true valuation is in $[0, 0.5]$. Assuming the other agent remains truthful, this will lead to the  first agent receiving all the items in her top $1/2$ quantile, and therefore a significantly increased utility of $3T/8$ ($\bbE[X_1 | X_1>0.5]/2$) compared to $T/3$ under truthful reporting. The social welfare however, goes down to $5T/8$ in this case.} 
Thus under the optimal policy, each agent has an incentive to misreport her valuations in order to gain individual utility. 
The incentives to misreport may be further amplified under an online learning based allocation algorithm that learns approximately optimal policies from the valuations observed in the previous rounds. In such settings, the agents can potentially mislead the online learning algorithm to \emph{learn a more favorable policy} over time by repeatedly misreporting their values. 

Motivated by these shortcomings, in this paper, we consider the problem of designing an online learning and allocation mechanism in the presence of \emph{strategic agents}. Specifically, we consider the problem of sequentially allocating $T$ items to $n$ strategic agents. The problem proceeds in $T$ rounds. 
In each round $t=1,\ldots, T$, 
the agents' true valuations $X_{i,t}, i=1,\ldots, n$ for the $t^{th}$ item are generated i.i.d. from a distribution $F$
\emph{a priori unknown} to the central planner.
However, the central planner can only observe a value $\tilde X_{i,t}$ reported by each agent $i$, which may or may not be the same as her true valuation $X_{i,t}$ for the item. Using the reported valuations from the current and previous rounds, the central planner needs to make an irrevocable decision of who to allocate the current item. The allocations should be made in a way such that  each agent at the end receives a fixed fraction $p^*_i$ of the $T$ items, where $p^*_i>0, 
\sum_{i=1}^n p^*_i=1$. The objective of the central planner is to maximize the total utility of the agents, where utility of each agent is defined as the sum total \emph{true} valuations of items received by the agent. 

Our main contribution is a mechanism that achieves  both: (a) \emph{Bayesian incentive compatibility}, i.e., assuming all the other agents are truthful, with high probability no single agent can gain a significant utility by deviating from the truthful reporting strategy, and; (b) \emph{near-optimal regret guarantees}, namely,  the utility of each individual agent under the online mechanism is ``close" to that achieved under the optimal offline allocation policy.    

\paragraph*{Organization}
After discussing the related literature in some further detail in Section \ref{sec: lit review}, we formally introduce the problem setting and some of the core
concepts in Section~\ref{sec: formulation}.  
Section~\ref{sec: algo and
main results} describes our online learning and allocation algorithm and provides formal statements of our main results (Theorem~\ref{thm: main BIC theorem} and Theorem~\ref{thm: individual regret}).
Section~\ref{sec: proof ideas} and Section~\ref{sec: proof outline} provide an overview of the proofs of the above theorems. All the missing details of the proofs are provided in the appendix. Finally, in Section~\ref{sec: future directions} we discuss some
limitations and future directions.


\section{Literature review}
\label{sec: lit review}
Our work lies at the intersection of online learning and mechanism design.
From an online learning perspective, our setting is closely related to the recent work on constrained online resource allocation under stochastic i.i.d. rewards/costs
(e.g., see \cite{DevanurSivan2019, agrawal2009dynamic,agrawal2014fast, balseiro2020best,balseiro2021regularized}).
However, a crucial assumption in those settings is that the central planner can observe the true rewards/costs of an allocation, which in our setting would mean that the central planner can observe agents' true valuations of the items being allocated. Our work extends these settings to allow for selfish and strategic agents who may have incentives to misreport their valuations. As discussed in the introduction, unless the online allocation mechanism design takes these incentives into account, selfish agents may significantly misreport their valuations to cause significant loss in social welfare.

Incentives and strategic agents have been previously considered in online allocation mechanism design, however, most of that work has focused on auction design where payments are used as a key mechanism for limiting rational agents' incentives to misreport their valuations. 
For example, \textcite{amin2014repeated} study a
posted-price mechanism in a repeated auction setting where buyers' valuations are context dependent. 
\textcite{golrezaei2019dynamic} extend this work to multi-buyer setting, using second-price auction with dynamic personalized
reserve prices.  \textcite{kanoria2021incentive} study a similar problem in a
non-contextual setting. 
(There is also a significant literature that studies learning
in repeated auction settings from the bidder's perspective.  Since our paper focuses on the central planner's point of view, we omit references to that literature. )
All of the above-mentioned works are concerned
with maximizing revenue for the seller, and use money/payments as a key instrument for
eliciting private information about the bidder's valuations for the items. In this paper, we are concerned with online allocation without money, and
the goal is to maximize each agent's utility.

Recently, there has been some work on studying reductions from mechanism design with money to those without money. \textcite{gorokh2021monetary} provided a black-box reduction from any
one-shot, BIC mechanism with money, to an
approximately BIC mechanism without money. 
However, their reduction relies crucially on knowing the true value distribution of agents and therefore is not applicable to our setting. 
\textcite{procaccia2013approximate} consider a specific (one-round) facility allocation problem and explicitly formulate the idea of designing
mechanisms without money to achieve approximately optimal performance against
mechanisms with money.  Subsequently, there is a series of papers that extended
the results on mechanism design without money in a single shot setting, when the
bidders' value distribution is unknown 
\cite{guo2010strategy,han2011strategy,cole2013mechanism,cole2013positive,cole2013mechanism}.
These papers either use a very restricted setting with just two agents, or use 
very specific/simple valuation functions for the agents. Even in these basic settings, they show that the best one can hope for is a constant approximation to what one can achieve with a mechanism that uses money. It is not clear what kind of regret guarantees such reductions imply in a repeated online learning setting. Therefore, we do not consider such reductions from auction mechanisms with money to be a fruitful direction for achieving our goals of both incentive compatibility and low (sublinear) regret for our online allocation problem. 


Finally, in a \emph{repeated} allocation settings with \emph{known} valuation distribution, there are more positive results for truthful mechanism design without money.
For example, \textcite{guo2009competitive} and later \textcite{balseiro2019multiagent}
studied the problem of repeatedly allocating items to agents with known value
distributions; both use a state-based ``promised utility'' framework.

To summarize, to the best of our knowledge, this is the first paper to incorporate strategic agents' incentives in the well-studied online allocation problem with stochastic i.i.d. rewards and unknown distributions. Thus, it bridges the gap between the online learning and allocation literature which focuses on non-strategic inputs, and the work on learning in repeated auctions which focuses on allocation mechanisms that  utilize money (payments) to achieve incentive compatibility.


\section{Problem formulation}
\label{sec: formulation}
\subsection*{The offline problem}
We first state the offline version of the problem which will serve as our
benchmark for the online problem. There is a set of $n$ agents, and a 
distribution $\Fvec$ over $\cX \coloneqq [0, \bar x]^n$.
\syremoved{\sacomment{Is assuming bounded values necessary? How does $\bar{x}$ appear in the theorems? Is this a common assumption?}}
Each draw
$\Xvec\sim\Fvec$ from this distribution represents the $n$ agents' valuations
of one item: $\Xvec = [X_1,\ldots,X_n]$. We assume that the agents' valuations are i.i.d., i.e.
$$\Fvec = F\otimes \ldots \otimes F.$$
A matching policy (aka allocation policy) maps, potentially with some exogenous randomness, a value vector $\Xvec$ to one of the agents $i\in \{1, \ldots, n\}$. 
Specifically, given a realized value vector
$\Xvec \in [0,\bar x]^n$, a (possibly randomized) policy $\pi$ maps $\Xvec$ to agent $\pi(\Xvec)\in \{1,\ldots, n\}$, with the probability of agent $i$ receiving an allocation given by $\bbP(\pi(\Xvec)=i)$. The offline optimization problem is to find a social welfare-maximizing policy $\pi^*$ such that each agent $i$ in expectation receives a predetermined fraction $p^*_i$ of the pool of items,  where $p_i^*>0, \sum_i p_i^*=1$. 
The problem of finding optimal policy can therefore be stated as the following
\begin{align}
    \max\limits_{\pi} \quad & \bbE\left[ \sum_{i=1}^nX_i \1(\pi(\Xvec)=i)\right]\label{prob: optimal partition problem}\\ 
    \text{ s.t.}\quad & \bbP(\pi(\Xvec)=i) = p^*_i\quad \forall i\nonumber
\end{align}
where the expectations are taken both over $\Xvec\sim \Fvec$ and any randomness in the mapping made by policy $\pi$ given $\Xvec$.
Solving the offline problem is non-trivial, as
it is an infinite dimensional optimization problem as stated in its' current
form in \eqref{prob: optimal partition problem}.
But it turns out to be closely related to Semi-Discrete Optimal Transport, and 
that the dual of \eqref{prob: optimal partition problem} can be written as
\begin{equation}
  \min_{\lambda\in\bbR^n} \cE(\lambda, \Fvec) \coloneqq \sum_{i\in[n]}\int_{\bbL_{i}(\lambda)} (X_i+\lambda_i)\,d\Fvec(\Xvec) - \lambda^\top p^*
  \label{prob: optimal transport dual}
\end{equation}
where $\bbL_{i}$ is what is sometimes referred to as the \emph{Laguerre cell}: 
\begin{equation}
\bbL_{i}(\lambda) = \left\{\Xvec: X_i + \lambda_i> X_j+\lambda_{j} \,\forall j\neq i, \right\}.
\label{eq: laguerre cell}
\end{equation}
Let $\lambda^*(\Fvec)$ denote an optimal solution to \eqref{prob: optimal
transport dual}. When it is clear from the context, we omit the distribution
$\Fvec$. It is known that an optimal solution to \eqref{prob: optimal partition
problem} is given by the following deterministic policy defined by the Laguerre cell partition (Proposition~2.1
\cite{aude2016stochastic}):
\begin{equation}
\label{eq:linearpartition}
    \pi^*(\Xvec) = i \text{ for all } \Xvec \in \bbL_i(\lambda^*), i=1,\ldots, n
\end{equation}
More generally, we will refer to any policy defined by a Laguerre cell partition as a greedy policy. 
\begin{definition}[Greedy allocation policy]
Consider any allocation policy that partitions the domain $[0,\bar x]^n$ as $\bbL_i(\lambda)$ (as defined in \eqref{eq:linearpartition}) for some $\lambda \in \bbR^n $. We refer to such a policy as the greedy allocation policy with  parameter $\lambda$.
\end{definition}
\syremoved{
\sacomment{Is such policy referred to as greedy policy in other literature or we are calling it greedy?} \sycomment{I'm calling it greedy.}
\sacomment{is $\lambda^*$ unique? should we say for some $\lambda$?}}

Note that there are efficient algorithms for solving \eqref{prob: optimal
transport dual} (see\cite{aude2016stochastic}) when the distribution $\Fvec$ is known. Therefore we will treat $\lambda^*(\cdot)$ as a black-box that can be computed efficiently for any given input distribution $\hat{\Fvec}$.

\subsection*{The online problem:  approximate Bayesian incentive compatibility and regret}
We are interested in the case when items are sequentially allocated over $T$ rounds,
and that the distribution $\Fvec$ is initially unknown. 
Specifically, in each round $t=1,\ldots, T$, 
the agents' true valuations $\Xvec_t =(X_{i,t}, i=1,\ldots, n)$ are generated i.i.d. from the distribution $\Fvec$
\emph{a priori unknown} to the central planner.
However, the central planner does not observe $\Xvec_t$ but only observes the reported valuations $\tilde \Xvec_t = (\tilde X_{i,t}, i=1,\ldots, n)$ which may or may not be the same as the true valuations.

An online allocation mechanism consists of a sequence of allocation policies $\pi_1, \ldots, \pi_t$ where the policy $\pi_t$ at time $t$ may be adaptively chosen based 
on the observed information until before time $t$:
\begin{equation}
    \cH_t = \{\tilde \Xvec_1,\ldots, \tilde \Xvec_{t-1}, \pi_1, \ldots,  \pi_{t-1}\}.
    \label{eq: history up to t}
\end{equation}
Given allocation policy $\pi_t$ at time $t$, the agent $i$'s utility at time $t$ is given by 
$$ u_i(\tilde \Xvec_t, \Xvec_t, \pi_t) = X_{i,t} \1[\pi_t(\tilde \Xvec_t) = i]$$
If an agent $i$ has already reached his target allocation of $p^*_i T$ items, then he cannot be allocated more items. Note that since the allocation policy may be randomized, for any given value vector $\Xvec$, $\pi_t(\Xvec)$ is a random variable. 
To ensure truthful reporting in presence of strategic agents, we are interested in mechanisms that are (approximately) Bayesian incentive compatible.

\begin{definition}[Approximate-BIC]
For an online allocation mechanism, let $\pi_t, t=1,\ldots, T$ be the sequence of allocations when all agents report truthfully, i.e., when $\tilde \Xvec_t = \Xvec_t, \forall t$; and let $\tilde \pi^i_t, t=1,\ldots, T$ be the sequence when all agents except $i$ report truthfully, i.e., $X_{j,t} = \tilde X_{j,t}, \forall j\ne i$.
Then the online allocation mechanism is called $(\alpha, \delta)$-approximate Bayesian Incentive Compatible if, for all $i=1,\ldots, n$, with probability at least $1-\delta$, 
$$
\sum_{t=1}^T u_i(\tilde\Xvec_t, \Xvec_t, \tilde \pi_t^i) - \sum_{t=1}^T u_i(\Xvec_t, \Xvec_t, \pi_t) \leq \alpha 
$$
Here, the probability is with respect to the randomness in true valuations $\Xvec_t\sim \Fvec$ and any randomness in the online allocation policy. For the online policy to be approximate-BIC, the statement should hold for all possible misreporting of valuations $\tilde X_{i,t} \ne X_{i,t}$.
\end{definition}
Therefore, if $\alpha$ is small, then an individual agent has little incentive to strategize. Note that approximate-BIC also implies that truthful reporting for all agents constitutes an approximate-Nash equilibrium.
\syremoved{
Now, assuming that all agents are truthful, we are also interested in the
difference between the social welfare achieved by the online allocation policy,
and the optimal welfare achieved in the offline expected problem: Let $\OPT$ be the
optimal objective of \eqref{prob: optimal partition problem}. 
That is, 
$$\OPT = \sum_{i=1}^n \mathbb{E}[u_i(\Xvec,\Xvec, \lambda^*(\Fvec))]$$
where $\lambda^*(\Fvec)$ denotes the corresponding greedy allocation policy. And, as before let $\pi_t$ be the sequence of allocation policies generated by our online allocation mechanism assuming truthful agents, then we define the total regret of the mechanism as
\begin{equation} 
  \label{eq: welfare regret}
  \Regret(T) = T\OPT - \sum_{i=1}^n \sum_{t=1}^T u_i(\Xvec_t, \Xvec_t, \pi_t)
\end{equation}}
Assuming that all agents are truthful, we are also interested in bounding each individual agent's regret.

\begin{definition}[Individual regret]
 We define an individual agent $i$'s regret under an online allocation mechanism as the difference between agent $i$'s realized utility over $T$ rounds and the expected 
utility achieved in the offline expected problem. That is,
\begin{equation} 
  \label{eq: individual regret}
  \Regret_i(T) =  T\mathbb{E}[u_i(\Xvec,\Xvec, \lambda^*)] - \sum_{t=1}^T u_i(\Xvec_t, \Xvec_t, \pi_t).
\end{equation}
Here $\pi_1,\ldots, \pi_T$ denote the allocation policies used by the online allocation mechanism in round $t=1,\ldots, T$. 
\end{definition}
Note that since social welfare is given by the sum of all agents' utilities, a bound on individual regret implies a bound on the regret in social welfare of the mechanism. 
\syremoved{Note that $\Regret(T) = \sum_i \Regret_i(T)$.
\sacomment{need a discussion on why $\OPT$ is a good benchmark}}

\section{Algorithm and main results}
\label{sec: algo and main results}
We present an online allocation mechanism that is approximately-BIC,  and further achieves low regret guarantees on individual regret when all agents are truthful.
Our algorithm contains two components: a learner, and a detector. Intuitively, the detector makes sure that the mechanism is approximately BIC, and the learner adaptively learns utility-maximizing allocation policies assuming truthful agents.

The \emph{learner}
runs in epochs with geometrically increasing lengths.
The starting time of each epoch $k$ is given by $L_k=2^k, k=0,1, \ldots$, which is also when the allocation policy is updated. At the end of each epoch (i.e., at time $t=L_k-1$ for epoch $k$), the learner takes all  the previously reported values from all the
agents, and uses them to construct an empirical distribution of the agents'
valuations. The learner implicitly assumes truthful agents in its
computations. Therefore, since the agents' true valuations are i.i.d., it first
constructs a single, one-dimensional empirical distribution $\hat F_t$, and then uses
it to construct the corresponding $n$-dimensional distribution $\hat \Fvec_t$:
\begin{equation}
    \label{eq:empirical}
\hat F_t(x) = \frac{1}{tn}\sum_{s=1}^t\sum_{i=1}^n \1[\tilde X_{i,s} \leq x]
\end{equation}
$$\hat \Fvec_t = \hat F_t\otimes \ldots \otimes \hat F_t$$
The learning algorithm then solves the offline problem \eqref{prob: optimal
transport dual} using $\hat\Fvec_t$, and uses the resulting greedy allocation policy characterized by $\lambda^*(\hat\Fvec_t)$ to allocate the items in the
following epoch.

In parallel to the learner, the \emph{detector} constructs and monitors, in each time step $t$,
and for each agent $i$, two empirical
distributions. One using the reported valuations from agent $i$: $\bar F_t$, and
one using the reported valuations from all the other agents: $\tilde F_t$. The detection
algorithm then computes the supremum between the two empirical CDFs, $\sup_x |\bar
F_t(x) - \tilde F_t(x)| $. If this difference is greater than a predetermined threshold $\Delta_t$,
then the detector raises a flag that there has been a violation of truthful reporting and the entire allocation game stops. 
Otherwise, the process continues. 

The threshold $\Delta_t$ needs to be chosen such that if everyone is truthful,
then with high probability the detection algorithm will not pull the trigger. At the same
time, if someone deviates from truthful reporting significantly, then it should detect this with
high probability. The typical concentration result used in comparing empirical
CDFs is the Dvoretzky-Kiefer-Wolfowitz (DKW) inequality\cite{dvoretzky1956asymptotic}.
However, since a strategic agent can adaptively change its'
misreporting strategy, we cannot directly apply the DKW inequality, which
assumes i.i.d. samples. Instead, we use martingale version of the DKW inequality (Lemma~\ref{lem: martingale
uniform convergence}), and use that to choose an appropriate threshold
$\Delta_t$.
The details are given in Algorithm~\ref{algo: main} and Algorithm~\ref{algo: detection}.

\SetKw{kwTerminates}{Terminates}
\SetKwInput{kwInput}{Input}
\SetKwInput{kwInitialize}{Initialize}
\SetKw{kwReturn}{Return}
\begin{algorithm}
\label{algo: main}
\caption{Epoch Based Online Allocation Algorithm}
\kwInput{$T, \delta$}
\kwInitialize{$\lambda = [0,\ldots, 0]$, $k=0$, $K=\log_2(T)$, $L_k=2^{k}, k=0,\ldots, K$\;}
\For{$t \gets 1,2,3,\ldots, T$}
{
    Observe $\Xvec_t$. 
    Run Detection Algorithm~(Algorithm~\ref{algo: detection}) with sample set
    $S=\{\Xvec_1,\ldots, \Xvec_{t}\}$, and threshold $\Delta_t = 64\sqrt{\frac{1}{t} \log(\frac{256et}{\delta})}$\\
    \If{Detection Algorithm \kwReturn Reject}{
      \kwTerminates.
    }
    \If{one of the agents $i\in \{1,\ldots, n\}$ has reached the allocation capacity $p^*_i T$}
    {Allocate randomly among agents who have not reached capacity}
    \Else{
    Allocate item using greedy allocation policy $\lambda$}
    
    \If{$t=L_{k+1}-1$}{
        Compute $\hat F_t$ from samples $\{\Xvec_1, \ldots, \Xvec_t\}$ as in \eqref{eq:empirical}. \\
        $\lambda \gets \lambda^*(\hat \Fvec_t)$\\
        $k\gets k+1$
    }
}
\end{algorithm}
\begin{algorithm}
  \label{algo: detection}
  \caption{Detection Algorithm}
  \kwInput{Sample set $S = \{\Xvec_1, \ldots, \Xvec_{t}\}$, threshold $\Delta_t$.}
  \For{$i \gets 1, \ldots, n$}{
      Compute $\bar F_t(x) = \frac{1}{t} \sum_{s=1}^t \1[\Xvec^s_{i} \leq x]$ as the empirical CDF of the samples collected from agent $i$\\
      Compute $\tilde F_t(x) = \frac{1}{t(n-1)} \sum_{s=1}^t\sum_{j\neq i} \1[\Xvec^s_j \leq x]$ be the empirical CDF of all reported values from the other agents.\\
      \If{$\sup_x |\tilde F_t(x) - \bar F_t(x) |\geq \frac{\Delta_t}{2}$}{
          \kwReturn Reject 
      }
  }
  \kwReturn Accept
\end{algorithm}

Our main results are the following guarantees on incentive compatibility and regret of our online allocation algorithm.

\begin{theorem}[Approximate-BIC]
  Algorithm~\ref{algo: main} is $(O(\sqrt{nT\log(nT/\delta)}), \delta)$-approximate BIC.
  \label{thm: main BIC theorem}
\end{theorem}
Since truthful reporting constitutes an
approximate equilibrium, it is reasonable to then assume that agents will act
truthfully. We show the following individual regret bound assuming truthfulness.
\begin{theorem}[Individual Regret]
  \label{thm: individual regret}
  Assuming all agents report their valuations truthfully, then under the online allocation mechanism given by Algorithm \ref{algo: main}, with probability $1-\delta$, every agent $i$'s
  individual regret can be bounded as:
  \begin{align*}
  \Regret_i(T) 
  & \leq \frac{4\sqrt{2}}{\sqrt[]{2}-1}\sqrt[]{nT\log(\frac{4n\log_2T+nT}{\delta})}\bar x\\
  & =  O(\sqrt{nT\log(nT/\delta)})
  \end{align*}
\end{theorem}
\syremoved{\sacomment{
Does that mean the bound on total regret (regret on total welfare) is $n^{3/2} \sqrt{T}$? Should we add that corollary? Is there an explanation for $n^{3/2}$? My guess (see later comment) is that this can be fixed, and it is because currently Claim 2 is too loose.
} \sycomment{I don't think claim 2 can be improved without additional assumption on the distributions. }}
Showing approximate incentive compatibility, and then guaranteeing
regret under the assumption of truthfulness, is an approach commonly seen in
the online mechanism design literature (e.g. Theorem~4 in
\cite{kanoria2021incentive}). 
In the next section, we describe the high level proof ideas for the main results above.

\section{Proof ideas}
\label{sec: proof ideas}
\paragraph{Proof ideas for Theorem~\ref{thm: main BIC theorem}}
We establish that the mechanism is approximately BIC by showing that no single agent has incentive to significantly deviate from reporting true valuations if all the other agents are truthful. The proof consists of two parts. In Step 1, we prove that any significant deviation from the truth can be detected and will lead the mechanism to terminate. In Step 2,3, we prove that in order to achieve a significant gain in utility, an agent indeed has to report values that significantly deviate from the truth.
\paragraph{Step 1} Assuming that there is only one (unidentified) strategic agent while all the other agents are truthful,  we first show that if the detector does not trigger a violation by time $t$
then with high probability, the empirical
distribution of valuations reported by the strategic agent is no more than
$O(1/\sqrt{t})$ away from the true distribution (Lemma~\ref{lem:
reported distributions have to be close to true distribution}). The key
observation here is that since the agents' valuations are i.i.d., we can compare
their reported values to detect if any single agent's distribution is
significantly different from everyone else's. A technical challenge in making statistical comparisons here is
that the strategic agent can adaptively change their reporting strategy over time based on
the realized outcomes. Therefore, we derive a novel martingale version of the DKW inequality to show concentration of the empirical distribution relative to the true underlying distribution.
\syremoved{The well known DKW inequality (Lemma~\ref{lem: dkw}) for comparing empirical distributions requires i.i.d. samples and therefore does not suffice for our purpose. We derive a novel martingale version of the DKW
inequality, using martingale uniform convergence (Lemma \ref{lem: martingale uniform convergence}).}
\paragraph{Step 2} In a given round, given the history, the mechanism's allocation policy is a fixed greedy allocation policy given by $\lambda$. If the distribution of strategic agent's reported values differs from the true distribution by at most $\Delta$,
 then the agent's expected utility gain in that round, compared to reporting truthfully, is at most
$O(\Delta)$,
(see Lemma~\ref{lem: utility gain bound given fixed allocation policy}).

\paragraph{Step 3} 
If over $t$ rounds, the  distribution of the strategic agent's reported values is at most $\Delta$ away from the true distribution,
then the learning
algorithm will, with high probability find an allocation policy that is at most 
$O(\sqrt{n}\Delta)$ away (in terms of individual utility) from what it would have learned if all the agents were
truthful instead (see Lemma~\ref{lem: utility bound from learning error}).

To understand the significance and distinction between results in Step 2 vs. Step 3, note that a strategic agent has two separate ways to gain utility.
The first is to report valuations in a way that the agent immediately wins more/better items under
the central planner's \emph{current} allocation policy. However, since
the central planner is updating its' allocation policy over time, the strategic
agent can also misreport in a way that benefits its' \emph{future utility},
by ``tricking'' the central planner into learning a policy that favors
him later on. Together, Step 2 and Step 3 show that the agent cannot gain significant advantage over being truthful in either manner.

In many existing works on online {\it auctions} mechanisms design, where the
central planner dynamically adjusts the reserve price over time, these two types
of strategic behaviors are in conflict: the agent either sacrifices
future utility to gain immediate utility; or  sacrifices near-term utility for
future utility. The results in those settings therefore often rely on this
observation to show approximate incentive compatibility.
In our case however, since there is no money involved, it is not clear if such a
conflict between short and long term utility exists. Nonetheless, we are able to
bound the agents' ability to strategize.  Step 2 bounds the agent's short term
incentive to be strategic, whereas Step 3 bounds the longer term incentive to be
strategic. Combining these steps gives us a proof for Theorem~\ref{thm: main
BIC theorem}.

\paragraph{Proof ideas for Theorem~\ref{thm: individual regret}}
Recall that here we assume all agents' are truthful. The proof involves two main steps.
\paragraph*{Step 1} We show that uniformly for any $t=1,\ldots,T$, with high probability, the empirical  distribution constructed from the first $t$ samples is close (within a distance of $\tilde O(1/\sqrt{nt})$) to the true  distribution $F$. Here, the factor of
$1/\sqrt{n}$ comes from the fact that in each round we observe $n$
independent samples from the value distribution, one from each of the agents. 
This also implies that if all the agents are reporting truthfully, then, with high probability, the detector  will not falsely trigger.
\paragraph*{Step 2} We show that the allocation policy learned under the empirical
distribution estimated from the samples is close to the the optimal allocation policy (Lemma~\ref{lem:
utility bound from learning error}). Specifically, after $t$ rounds if the empirical CDF is
at most $O(1/\sqrt{nt})$ away from the true distribution, then each agents'
expected utility in one round under the learned allocation policy is at most $\sqrt{n/t}$ away (both from above and from below) from the optimal.
By using an epoch based approach
we can then show that each agents' individual regret is with high probability bounded by
$O(\sqrt{nT})$ over the entire planning horizon $T$.

\section{Proof details}
\label{sec: proof outline}
We will now outline our proof in more detail. All missing proofs can be found in the Appendix. First we state the following martingale variation of the well-known Dvoretzky-Kiefer-Wolfowitz (DKW) inequality\cite{dvoretzky1956asymptotic}.
This is critical when dealing with strategic agents as they can adapt their strategy over time, resulting in non-independent (reported) values.
\begin{restatable}[Martingale Version of DKW Inequality]{lemma}{martingaleuniformconvergence}
  \label{lem: martingale uniform convergence}
  Given a sequence of random variables $Y_1,\ldots, Y_T$, let $\cF_t=\sigma(Y_1,\ldots, Y_t), t=1,\ldots, T$ be the filtration representing the information in the first $t$ variables. Let \mbox{$F_t(y) := \Pr(Y_t \le y|\cF_{t-1})$}, and  $\bar F_T(y) := \frac{1}{T}\sum_{t=1}^T \1[Y_t\leq y]$.
  Then, 
  $$\bbP\left(\sup_y \left|\bar F_T(y) - \frac{1}{T}\sum_{t=1}^T F_t(y)\right|\geq \alpha\right) \leq \left(\frac{128eT}{\alpha}\right)e^{-T\alpha^2/128}$$
\end{restatable}
\syremoved{Proof of Lemma~\ref{lem: martingale uniform convergence} is in
Appendix~\ref{sec: proof of martingale uniform convergence}. }

Next we introduce a new notation  to to denote the fraction of allocation that $j$ receives under the greedy allocation policy with parameter $\lambda$ and valuation distribution $\Fvec$:
$$p_j(\Fvec, \lambda) := \bbP_{\Xvec\sim \Fvec}(\Xvec\in \bbL_j(\lambda)).$$
We start with proving Theorem~\ref{thm: individual regret}, as we will use this to prove Theorem~\ref{thm: main BIC theorem} later.

\subsection{Individual Regret Bound (Theorem~\ref{thm: individual regret})}
In Algorithm~\ref{algo: main}, the allocation policy is trained on the
empirical distribution constructed from samples. We want to show that this difference between empirical and population distribution will not impair the performance of the resulting allocation policy too much.

\syremoved{\sacomment{we can move the proof of the lemma below to appendix if needed.}}
\begin{lemma}
  \label{lem: utility bound from learning error}
  Let $\Gvec = G^1 \otimes \ldots
  \otimes G^n$, and $\Fvec = F^1\otimes \ldots \otimes F^n$ be two distributions over $[0,\bar x]^n$ where the marginals on each
  coordinate are independent. Suppose $\sup_x|F^i(x)-G^i(x)|\leq
  \Delta\,\forall i$.  Let $\lambda = \lambda^*(\Gvec)$, and $\lambda^* =
  \lambda^*(\Fvec)$.  Then 
  $$
  \left|\bbE_{\Xvec\sim \Fvec}[u_i(\Xvec, \Xvec, \lambda)] - \bbE_{\Xvec\sim \Fvec}[u_i[\Xvec, \Xvec, \lambda^*]]\right|\leq n\Delta\bar x
  $$
\end{lemma}

\paragraph*{Proof of Theorem~\ref{thm: individual regret}} Now we have the main ingredients for Theorem~\ref{thm: individual regret}. 
We use the DKW inequality to show that the empirical distribution constructed in \eqref{eq:empirical} is close to the
true distribution w.h.p.. Then we use  Lemma~\ref{lem: utility bound from learning error} to show that the allocation policy selected by the
learner based on the empirical distribution is almost optimal in expectation. The details can be found in the Appendix~\ref{sec: proof of ind regret theorem}.

\subsection{Approximate-Bayesian Incentive Compatibility (Theorem~\ref{thm: main BIC theorem})}
Theorem~\ref{thm: individual regret} says that online utility of each agent
cannot be too far below the offline optimum if everyone behaves truthfully. In
order to show approximate-BIC, it suffices to show that the
strategic agent cannot gain too much more than the offline optimum.
To do so, we need to bound both the short term and longer term incentives for
the agent to be strategic. 

\subsubsection{Short term incentive}
We start with bounding the short term strategic incentive.  We first show that
if agent reports from an average distribution that is very different from the
truthful distribution, then with high probability Algorithm~\ref{algo:
detection} can detect that. Note that given the strategic agent's strategy in a given round, 
his reported value is drawn from a distribution potentially different from $F$.

\begin{lemma}
  \label{lem: reported distributions have to be close to true distribution}
  Fix a time step $t$. Let $\Delta = 64\sqrt{\frac{\log(\frac{256et}{\delta})}{t}}$. 
  Let $F_s, s=1,\ldots, t$ be the strategic agent's
  reported value distributions in each time step given the history, i.e.,
  $F_s(x) := \bbP(\tilde X_{i,s} \le x | {\cal H}_s).$
  If the average distribution $\bar F = \frac{1}{t}\sum_{s=1}^t F_s$ is
  such that $\sup_x|\bar F(x) - F(x)| \geq \Delta$, then Algorithm~\ref{algo:
  main} will terminate at or before time $t$ with probability at least $1-\delta$.
\end{lemma}
\syremoved{
\sacomment{The discussion below is not rigorous in its definitions. I will try to rewrite it once I go through the proof of Lemma 5. I am not even sure why there is a reporting function. You have two random variables $X, \tilde X$, and it seems you want to couple them so that the marginal distribution distribution is same as that of $X, \tilde X$ but the joint distribution is such that $\Pr(X\ge \tilde X)=0$. I tentatively rewrote the lemma statement below.}}

Next, we show that if the agent restricts the reported distribution to not deviate more than $\Delta$ from the true distribution (so that the deviation may go undetected by the detection algorithm), then the potential gain in the agent's utility compared to truthful reporting is upper bounded by $\bar x\Delta$. This bounds the agent's incentive to be strategic.
\begin{lemma}
  \label{lem: utility gain bound given fixed allocation policy}
  Fix a round $t$ and a single strategic agent $i$, so that the remaining agents are truthful, i.e., $\tilde X_{j,t}=X_{j,t}, \forall j\ne i$. Let $F_r(\cdot)$ denote the marginal distribution of values $\tilde X_{i,t}$ reported by the strategic agent $i$ at time $t$ conditional on the history, i.e., 
$$F_r(x) := \bbP(\tilde X_{i,t} \le x | {\cal H}_t).$$  Suppose that $\sup_x |F(x) - F_r(x)| \leq \Delta$. Then, at any time $t$, given any  greedy allocation policy $\lambda$,  
  $$\bbE[u_i(\tilde \Xvec_t, \Xvec_t, \lambda) | {\cal H}_t] - \bbE[u_i(\Xvec_t, \Xvec_t,
  \lambda)] \leq \bar x\Delta$$
\end{lemma}
Note that  $F_r$ specifies only the marginal distribution of $\tilde X_{i,t} | \cH_t$ and not the joint distribution of $(\tilde \Xvec_t, \Xvec_t)|\cH_t$. Indeed the above lemma claims that the given bound on utility gain holds for all possible joint distributions as long as the marginal $F_r$ of $\tilde X_{i,t}|\cH_t$ is at most $\Delta$ away from $F$.

Intuitively, Lemma~\ref{lem: reported distributions have to be close to true
distribution} and Lemma~\ref{lem: utility gain bound given fixed allocation policy}
together bound the agent's short term incentive to be strategic: if the agent deviates
from the truthful distribution too much, then the mechanism will terminate early and the agent will lose out on all the future utility (Lemma~\ref{lem: reported distributions
have to be close to true distribution}); 
and given any greedy allocation strategy set by the central planner, we have that 
if the agents deviates within the undetectable
range of Algorithm~\ref{algo: detection}, then the gain in utility compared to acting truthfully is small (Lemma~\ref{lem: utility gain bound given
fixed allocation policy}). Next, we bound an agent's incentive to lie in order to make the mechanism learn a suboptimal greedy allocation policy that is more favorable to the agent.

\subsubsection{Long term incentive}
In order to bound the longer term incentive to be strategic, we want to show that
despite agent $i$ being strategic, the central planner can still learn an allocation policy that closely approximates the offline optimal allocation policy. 
This means that the agent's influence over 
the central planner's allocation policies is limited. 
\begin{lemma}
  \label{lem: learning under strategic report}
  Fix a round $T'\le T$ and a strategic agent $i$. If agent $i$ is the only one being strategic, and Algorithm~\ref{algo:
  detection} has not been triggered by the end of time $T'$, then with
  probability $1-\delta$, $\hat\lambda\coloneqq \lambda^*(\hat \Fvec_{T'})$ satisfies
  $$
  \bbE[u_i(\Xvec, \Xvec, \hat\lambda)]-\bbE[u_i(\Xvec, \Xvec, \lambda^*)] \leq n\Delta_{T'} \bar x
  $$
  where $\Delta_{T'} = 81\sqrt{\frac{1}{nT'} \log(\frac{256e(T')}{\delta})}$ and $\lambda^*=\lambda^*(\Fvec)$.
\end{lemma}
In particular, consider $T'=L_k$. Then the lemma above shows that if an agent was not kicked out by the
end of epoch $k-1$, then with high probability the greedy allocation policy  in epoch $k$ will be such that the agents' expected
utility by being truthful is close to what he would have received in the offline
optimal solution (Lemma~\ref{lem: learning under strategic report}). We can now
combine this result with Lemma~\ref{lem: reported distributions have to be close
to true distribution} and Lemma~\ref{lem: utility gain bound given fixed
allocation policy} to bound the utility that any single strategic agent can gain
over the entire trajectory. 
\begin{lemma}
  If agent $i$ is the only one being strategic, then with probability
  $1-\delta$, agent $i$'s online utility is upper bounded by
  $$
  \sum_{t=1}^T u_i(\tilde \Xvec_t, \Xvec_t, \tilde \lambda_{k_t}) \leq T\bbE\left[u_i(\Xvec, \Xvec, \lambda^*)\right] + \frac{286\sqrt{2}}{\sqrt{2}-1}\sqrt{nT \log(\frac{256e\log_2T}{\delta})}  \bar x
  $$
  Here $k_t$ denotes the epoch number that time step $t$ lies in, and $ \tilde \lambda_{k_t}$ denotes the allocation policy used by the central planner in that epoch.
  \label{lem: single strategic agent utility bound}
\end{lemma}

\paragraph*{Proof of Theorem~\ref{thm: main BIC theorem}}
We have already proven Theorem \ref{thm: individual regret} which bounds individual regret defined as  the difference $T\mathbb{E}[u_i(\Xvec,\Xvec, \lambda^*)] - \sum_{t=1}^T u_i(\Xvec_t, \Xvec_t, \lambda_{k_t})$, i.e., the difference between the utility of agent $i$ under the offline optimal policy and that under the allocation policy learned by the algorithm {\it when all the agents are truthful}. The proof of Theorem \ref{thm: main BIC theorem} follows from plugging in the upper bound on $T\mathbb{E}[u_i(\Xvec,\Xvec, \lambda^*)]$ from this theorem into Lemma~\ref{lem: single strategic agent utility bound}, to obtain the desired bound on the expression $  \sum_{t=1}^T u_i(\tilde \Xvec_t, \Xvec_t, \tilde \lambda_{k_t}) - \sum_{t=1}^T u_i(\Xvec_t, \Xvec_t, \lambda_{k_t})$, i.e., on the total gain in utility achievable by misreporting under our mechanism. Further details are in Appendix~\ref{sec: proof of main bic theorem}.

\section{Limitations and Future Directions}
\label{sec: future directions}
\syremoved{We believe that online resource allocation with strategic agents is an important
research question and we are excited to be making progress in this direction.} 
Although our goal is to develop mechanisms that are robust to selfish and strategic agents, real applications often involve bad faith actors that have extrinsic motivation to behave adversarially. As such, deployment of such resource allocation mechanisms to critical applications requires significant additional validation. 
In future work we would like to explore the limit of relaxing the i.i.d.
assumption that we place on the distribution of valuations across agents. This is a natural relaxation
because if one agent thinks the item is good then it's likely that other agents
would like the item as well. Furthermore it is also conceivable that agents are
heterogeneous and so have different value distributions for the items. However
this seems to require a completely different strategy for detecting, and
disincentivize strategic behaviors, as we can no longer catch the strategic
agent through comparing each agent's reported distribution with that of others.

\section*{Acknowledgement}
This work was supported in part by an NSF CAREER award [CMMI 1846792] awarded to author S. Agrawal.
\printbibliography

@article{amin2014repeated,
  title={Repeated contextual auctions with strategic buyers},
  author={Amin, Kareem and Rostamizadeh, Afshin and Syed, Umar},
  journal={Advances in Neural Information Processing Systems},
  volume={27},
  year={2014}
}

@article{golrezaei2019dynamic,
  title={Dynamic incentive-aware learning: Robust pricing in contextual auctions},
  author={Golrezaei, Negin and Javanmard, Adel and Mirrokni, Vahab},
  journal={Advances in Neural Information Processing Systems},
  volume={32},
  year={2019}
}

@article{kanoria2021incentive,
  title={Incentive-compatible learning of reserve prices for repeated auctions},
  author={Kanoria, Yash and Nazerzadeh, Hamid},
  journal={Operations Research},
  volume={69},
  number={2},
  pages={509--524},
  year={2021},
  publisher={INFORMS}
}

@article{procaccia2013approximate,
  title={Approximate mechanism design without money},
  author={Procaccia, Ariel D and Tennenholtz, Moshe},
  journal={ACM Transactions on Economics and Computation (TEAC)},
  volume={1},
  number={4},
  pages={1--26},
  year={2013},
  publisher={ACM New York, NY, USA}
}

@inproceedings{guo2010strategy,
  title={Strategy-proof allocation of multiple items between two agents without payments or priors.},
  author={Guo, Mingyu and Conitzer, Vincent},
  booktitle={AAMAS},
  pages={881--888},
  year={2010},
  organization={Citeseer}
}

@inproceedings{han2011strategy,
  title={On strategy-proof allocation without payments or priors},
  author={Han, Li and Su, Chunzhi and Tang, Linpeng and Zhang, Hongyang},
  booktitle={International Workshop on Internet and Network Economics},
  pages={182--193},
  year={2011},
  organization={Springer}
}

@article{cole2013positive,
  title={Positive results for mechanism design without money},
  author={Cole, Richard and Gkatzelis, Vasilis and Goel, Gagan},
  year={2013}
}

@inproceedings{cole2013mechanism,
  title={Mechanism design for fair division: allocating divisible items without payments},
  author={Cole, Richard and Gkatzelis, Vasilis and Goel, Gagan},
  booktitle={Proceedings of the fourteenth ACM conference on Electronic commerce},
  pages={251--268},
  year={2013}
}

@inproceedings{guo2009competitive,
  title={Competitive repeated allocation without payments},
  author={Guo, Mingyu and Conitzer, Vincent and Reeves, Daniel M},
  booktitle={International Workshop on Internet and Network Economics},
  pages={244--255},
  year={2009},
  organization={Springer}
}

@article{balseiro2019multiagent,
  title={Multiagent mechanism design without money},
  author={Balseiro, Santiago R and Gurkan, Huseyin and Sun, Peng},
  journal={Operations Research},
  volume={67},
  number={5},
  pages={1417--1436},
  year={2019},
  publisher={INFORMS}
}

@article{gorokh2021monetary,
  title={From monetary to nonmonetary mechanism design via artificial currencies},
  author={Gorokh, Artur and Banerjee, Siddhartha and Iyer, Krishnamurthy},
  journal={Mathematics of Operations Research},
  year={2021},
  publisher={INFORMS}
}

@article{dvoretzky1956asymptotic,
  title={Asymptotic minimax character of the sample distribution function and of the classical multinomial estimator},
  author={Dvoretzky, Aryeh and Kiefer, Jack and Wolfowitz, Jacob},
  journal={The Annals of Mathematical Statistics},
  pages={642--669},
  year={1956},
  publisher={JSTOR}
}

@article{rakhlin2015sequential,
  title={Sequential complexities and uniform martingale laws of large numbers},
  author={Rakhlin, Alexander and Sridharan, Karthik and Tewari, Ambuj},
  journal={Probability Theory and Related Fields},
  volume={161},
  number={1-2},
  pages={111--153},
  year={2015},
  publisher={Springer}
}

@article{aude2016stochastic,
  title={Stochastic optimization for large-scale optimal transport},
  author={Aude, Genevay and Cuturi, Marco and Peyr{\'e}, Gabriel and Bach, Francis},
  journal={arXiv preprint arXiv:1605.08527},
  year={2016}
}

@article{agrawal2009dynamic,
author = {Agrawal, Shipra and Wang, Zizhuo and Ye, Yinyu},
title = {A Dynamic Near-Optimal Algorithm for Online Linear Programming},
year = {2014},
issue_date = {August 2014},
publisher = {INFORMS},
address = {Linthicum, MD, USA},
volume = {62},
number = {4},
issn = {0030-364X},
journal = {Operations Research},
month = {Aug},
pages = {876–890},
numpages = {15}
}

@article{DevanurSivan2019,
author = {Devanur, Nikhil R. and Jain, Kamal and Sivan, Balasubramanian and Wilkens, Christopher A.},
title = {Near Optimal Online Algorithms and Fast Approximation Algorithms for Resource Allocation Problems},
year = {2019},
issue_date = {February 2019},
publisher = {Association for Computing Machinery},
address = {New York, NY, USA},
volume = {66},
number = {1},
issn = {0004-5411},
url = {https://doi.org/10.1145/3284177},
doi = {10.1145/3284177},
journal = {J. ACM},
month = {jan},
articleno = {7},
numpages = {41}
}

@inproceedings{agrawal2014fast,
  title={Fast algorithms for online stochastic convex programming},
  author={Agrawal, Shipra and Devanur, Nikhil R},
  booktitle={Proceedings of the twenty-sixth annual ACM-SIAM symposium on Discrete algorithms},
  pages={1405--1424},
  year={2014},
  organization={SIAM}
}

@article{balseiro2020best,
  title={The best of many worlds: Dual mirror descent for online allocation problems},
  author={Balseiro, Santiago and Lu, Haihao and Mirrokni, Vahab},
  journal={arXiv preprint arXiv:2011.10124},
  year={2020}
}

@inproceedings{balseiro2021regularized,
  title={Regularized online allocation problems: Fairness and beyond},
  author={Balseiro, Santiago and Lu, Haihao and Mirrokni, Vahab},
  booktitle={International Conference on Machine Learning},
  pages={630--639},
  year={2021},
  organization={PMLR}
}

\newpage
\appendix

\section{Concentration Results}
\begin{lemma}[DKW Inequality \cite{dvoretzky1956asymptotic}]
  \label{lem: dkw}
  Given i.i.d. samples $X_1, \ldots, X_T$ from a distribution $F$ (cdf), 
  let $\hat F_T(x) = \frac{1}{T}\sum_{t=1}^T \1[X_t\leq x]$.
  Then, 
  $$\bbP\left(\sup_x \left|\hat F_T(x)-  F(x)\right|\geq \alpha\right) \leq 2e^{-2T\alpha^2}$$
\end{lemma}
\subsection{Proof of Lemma~\ref{lem: martingale uniform convergence}}
\label{sec: proof of martingale uniform convergence}
\martingaleuniformconvergence*
\begin{proof}
  This follows from sequential uniform convergence, see Lemma~10,11 in
  \cite{rakhlin2015sequential}, and the fact that indicator functions have
  fat-shattering dimension $1$. 
\end{proof}
A more convenient way to use Lemma~\ref{lem: martingale uniform convergence} is the following corollary:
\begin{corollary}
  \label{cor: martingale uniform convergence}
Given a sequence of random variables $Y_1,\ldots, Y_T$, let $\cF_t=\sigma(Y_1,\ldots, Y_t), t=1,\ldots, T$ be the filtration representing the information in the first $t$ variables. Suppose \mbox{$F_t(y) = \Pr(Y_t \le y|\cF_{t-1})$}, and let $\bar F_T(y) := \frac{1}{T}\sum_{t=1}^T \1[Y_t\leq y]$.
  If $\alpha\geq 16\sqrt{\frac{ \log(\frac{128 et}{\delta})}{T}}$, then with probability $1-\delta$
  $$\sup_x \left|\bar F_T(x) - \frac{1}{T}\sum_{t=1}^T F_t(x)\right|\leq \alpha$$
\end{corollary}
\begin{proof}
  \begin{align*}
    &\left(\frac{128eT}{\alpha}\right)e^{-T\alpha^2/128} \leq\delta\\
  \iff& \alpha^2\geq \frac{128 \log(\frac{128 eT}{\delta})}{T} + \frac{128}{T}\log(\frac{1}{\alpha^2})\\
  \impliedby& \alpha^2\geq \frac{256 \log(\frac{128 eT}{\delta})}{T} \\
  \iff& \alpha\geq 16\sqrt{\frac{ \log(\frac{128 eT}{\delta})}{T}}
\end{align*}
\end{proof}

\section{Proof of Theorem~\ref{thm: individual regret}}
\label{sec: proof of ind regret theorem}
\subsection{Proof of Lemma~\ref{lem: utility bound from learning error}}
We first state two helper claims. The first one states that for any fixed greedy allocation policy $\lambda$, if the two
distributions of valuations are similar, then the final allocation sizes for
each receiver will also be close.
\begin{claim}
  Fix a greedy allocation policy $\lambda$. Let $\Gvec = G^1 \otimes \ldots
  \otimes G^n$, and $\Fvec = F^1\otimes \ldots \otimes F^n$ be two
  distributions over $[0, \bar x]^n$ where the marginals in each coordinate are
  independent.  Suppose $\sup_x|F^i(x) - G^i(x)|\leq \Delta\,\forall i$.  Then 
  $$\sum_j (p_j(\Fvec,
  \lambda)-p_j(\Gvec, \lambda))^+ = \sum_j (p_j(\Gvec, \lambda)-p_j(\Fvec, \lambda))^+ \leq n\Delta.$$
  \label{claim: misreporting size gain abstract}
\end{claim}
\syremoved{Proof of this Claim is in Appendix~\ref{sec: proof of claim misreporting size
gain abstract}. }
Next we show that if the allocation sizes are similar for two
different greedy allocation policies, then the corresponding allocation decisions (domain partitions)  are also similar.
\begin{claim}
\label{claim: misreporting allocation difference}
  Let $\lambda', \lambda$ be any two fixed greedy allocation policies, 
  and $\Fvec$ a distribution over $[0, \bar x]^n$.  For
  all $j$, let $\Omega_j'= \bbL_j(\lambda')$, and
  $\Omega_j= \bbL_j(\lambda)$.  Suppose $\sum_j
  (p_j(\Fvec, \lambda')-p_j(\Fvec, \lambda))^+ = \sum_j (p_j(\Fvec, \lambda)-p_j(\Fvec, \lambda'))^+ \leq
  \Delta$. Then
  $$\bbP(\Omega'_j \setminus \Omega_j) \leq \Delta \quad\forall j$$
\end{claim}

\syremoved{\sacomment{From the proof it looks like the above bound is quite loose if used for proving for "all $j$". What I mean is that I think one could use the same proof to show that $\Pr(\cup_j \Omega'_j \backslash \Omega_j )\le \Delta$ as well. This is why I think bounding total regret using individual regret bound gets us an extra $n$. }}
Using these two Claims, we can now prove Lemma~\ref{lem: utility bound from learning error}. The proofs for these two helper Claims follow after the proof of Lemma~\ref{lem: utility bound from learning error}.
\begin{proof}[Proof of Lemma~\ref{lem: utility bound from learning error}]
  Claim~\ref{claim: misreporting size gain abstract} shows that 
  $$
  \sum_j (p_j(\Fvec,
  \lambda)-p_j(\Gvec, \lambda))^+ = \sum_j (p_j(\Gvec, \lambda)-p_j(\Fvec,
  \lambda))^+ \leq n\Delta
  $$
  Note that by definition of $\lambda^*$ (due to the constraint $\Pr(\Omega_j)=p_j^*$ in problem \eqref{prob: optimal partition problem}), we have  $p_j(\Gvec,\lambda) = p^*_j= p_j(\Fvec, \lambda^*)$ for all $j$. This means that 
  $$
  \sum_j (p_j(\Fvec,
  \lambda)-p_j(\Fvec, \lambda^*))^+ = \sum_j (p_j(\Fvec, \lambda^*)-p_j(\Fvec,
  \lambda))^+ \leq n\Delta
  $$
  Now we can apply Claim~\ref{claim: misreporting allocation difference} and
  conclude that
  \begin{align}
    \bbP(\Omega^*_j \setminus \Omega_j) \leq n\Delta \quad\forall j,\quad 
    \text{and }\quad &\bbP(\Omega_j \setminus \Omega^*_j) \leq n\Delta \quad\forall j,\label{eq: utility bound from learning error helper}
  \end{align}
  where $\Omega_j = \bbL_j(\lambda)$ and
  $\Omega^*_j = \bbL_j(\lambda^*)$.
  Therefore,
  \begin{align*}
    &\bbE_{\Xvec\sim \Fvec}[u_i(\Xvec, \Xvec, \lambda)] - \bbE_{\Xvec\sim \Fvec}[u_i[\Xvec, \Xvec, \lambda^*]]\\
    = &\int_{\Xvec\in \Omega_i}  X_i\,d\Fvec(\Xvec) - \int_{\Xvec\in \Omega^*_i}  X_i\,d\Fvec(\Xvec)\\
    \leq &\int_{\Xvec\in \Omega_i\setminus\Omega^*_i}  X_i\,d\Fvec(\Xvec) \\
    \leq &n\Delta \bar x
  \end{align*}
  Using the same steps as above we can also show that 
  $$
  \bbE_{\Xvec\sim \Fvec}[u_i(\Xvec, \Xvec, \lambda^*)] - \bbE_{\Xvec\sim \Fvec}[u_i[\Xvec, \Xvec, \lambda]]\leq n\Delta\bar x.
  $$

\end{proof}

\subsubsection{Proof of Claim~\ref{claim: misreporting size gain abstract}}
\label{sec: proof of claim misreporting size gain abstract}
We first show variant of Claim~\ref{claim: misreporting size gain abstract} where
the two distributions only differ in one coordinate:
\begin{claim}
  Fix a greedy allocation policy $\lambda$. Let $\Gvec = G^1 \otimes G^2 \ldots
  \otimes G^n$, and $\Fvec = F^1\otimes F^2\ldots \otimes F^n$ be two
  distributions over $[0, \bar x]^n$ where the marginals in each coordinate are
  independent. Assume that $\Gvec$ and $\Fvec$ differ only in one coordinate, w.l.o.g. say coordinate $i$. 
  Then, if $\sup_x|F^i(x) - G^i(x)|\leq \Delta$,  
  $$\sum_j (p_j(\Fvec, \lambda)-p_j(\Gvec, \lambda))^+ = \sum_j (p_j(\Gvec,
  \lambda)-p_j(\Fvec, \lambda))^+ \leq \Delta.$$
  \label{claim: misreporting size gain abstract 0}
\end{claim}
\begin{proof}
  We start with the distribution $\Gvec =  F^1 \cdots F^{i-1} \otimes G^i \otimes F^{i+1} \cdots F^n$ 
  and 
  replace $G^i$ with a distribution $F^i$ to construct $\Fvec=F^1\otimes F^2 \cdots \otimes F^n$. 
  We will construct $F^i$ in such a way that it is at most $\Delta$ away from $G^i$ and the 
  changes in the allocation proportions are maximized. Note that since 
  $\sum_i p_i(\Fvec, \lambda) = 1$ and $\sum_i p_i(\Gvec, \lambda) = 1$, we know that
  $$LHS(\Fvec)\coloneqq\sum_j (p_j(\Fvec, \lambda)-p_j(\Gvec, \lambda))^+ = \sum_j (p_j(\Gvec,
  \lambda)-p_j(\Fvec, \lambda))^+\coloneqq RHS(\Fvec) $$
  is always true for any $\Fvec, \Gvec, \lambda$. This means that we can focus
  on either maximizing either the LHS or the RHS of the above equation.  There
  are two types of $F^i$ that we can use. One is such that $p_i(\Fvec, \lambda)
  - p_i(\Gvec, \lambda)\geq 0$ and the other is $p_i(\Fvec, \lambda) -
  p_i(\Gvec, \lambda)< 0$. We can therefore bound the above quantity under these
  two scenarios separately:
  \begin{align}
    \max_{F^i: p_i(\Fvec, \lambda) - p_i(\Gvec, \lambda)\geq 0} RHS(\Fvec) \iff 
    \max_{F^i: p_i(\Fvec, \lambda) - p_i(\Gvec, \lambda)\geq 0} \sum_{j\neq i} (p_j(\Gvec,
    \lambda)-p_j(\Fvec, \lambda))^+\label{eq: claim misreporting size gain helper 1}\\
    \max_{F^i: p_i(\Fvec, \lambda) - p_i(\Gvec, \lambda)< 0} LHS(\Fvec) \iff 
    \max_{F^i: p_i(\Fvec, \lambda) - p_i(\Gvec, \lambda)< 0} \sum_{j\neq i} (p_j(\Fvec, \lambda) - p_j(\Gvec,
    \lambda))^+\label{eq: claim misreporting size gain helper 2}
  \end{align}
  Therefore for the rest of the proof we can focus on bounding the right hand
  side of \eqref{eq: claim misreporting size gain helper 1} and \eqref{eq: claim
  misreporting size gain helper 2}.  
  \syremoved{Note that since agents' valuations are independent, the allocation proportions for each receiver can be written in the following form:
  \begin{equation}
  p_i(\Fvec, \lambda) = \int_x \prod_{j\neq i} F^j(x + \lambda_i-\lambda_j)dF^i(x),
  \label{eq: winning prob agent i}
  \end{equation}
  \begin{equation}
  p_j(\Fvec, \lambda) = \int_x F^i(x)\prod_{k\not\in\{i,j\}} F^k(x +
  \lambda_j-\lambda_k)dF^j(x).
  \label{eq: winning prob agent j}
  \end{equation}
  }

  \paragraph*{Bounding the RHS of \eqref{eq: claim misreporting size gain helper 1}}
  Let $\tilde F(x)\coloneqq (G^i(x)-\Delta)^+ \,\forall x<\bar x, \tilde F(\bar
  x)\coloneqq 1$.  We claim that the $F^i$ that maximizes $\sum_{j\neq i}
  (p_j(\Gvec, \lambda)-p_j(\Fvec, \lambda))^+ $, \newshipra{while being at most $\Delta$ away}, is $\tilde F$. To see this,
  consider a different distribution $F'$ on the support $[0,\bar x]$ such that $\sup_x|F'(x) - G^i(x)|\leq \Delta$. We know that $F'(x) \geq \tilde F(x)$. 
  
  Later in Claim \ref{claim: generalized monotone mapping between two distributions}, we show that for any two distributions $G$ and $F$, we can sample $X\sim F$ using $Y$ sampled from $G$ by performing the following transformation:
  $$F^{-1}(G^u(Y))$$
  where $G^u$ is the random variable defined for distribution $G$ in \eqref{eq:
smoothed F CDF} and $F^{-1}\coloneqq \inf\{x\in \bbR : F(x)\geq p\}$ denotes the generalized inverse, sometimes also referred to as the quantile function. This is essentially the inverse CDF method applied to a general distribution (instead of a uniformly sampled variable).
  \syremoved{there exists a unique joint distribution $r$ with ($G$ and $F$ as its' marginals) such that the conditional distribution $r(\cdot | Y)$ has the following
\emph{monotonicity} property: define $\upperx_r(\cdot), \lowerx_r(\cdot)$ so that $X \in [\lowerx_r(Y), \upperx_r(Y)]$ almost surely, i.e.,
\begin{align*}
  \upperx_r(y) = \inf\{x: \bbP(X>x | Y=y) = 0\}\\
  \lowerx_r(y) = \sup\{x: \bbP(X < x| Y=y) = 0\},
\end{align*}
then
$$\upperx_r(y_1) \leq \lowerx_r(y_2) \quad\forall y_1<y_2.$$
In particular, the random variable $X|Y \sim r(\cdot|Y)$ can be sampled as 
$F^{-1}(G^u(Y))$, where $G^u$ is the random variable defined for distribution $G$ in \eqref{eq:
smoothed F CDF} and $F^{-1}\coloneqq \inf\{x\in \bbR : F(x)\geq p\}$ denotes the generalized inverse, sometimes also referred to as the quantile function. 
}
  In particular, let $G^{iu}$ be the following random function:
  \begin{equation*}
    G^{iu}(y) = 
    \begin{cases}
      G^i(y) &\text{ if } G^i(y) = G^i(y_-)\\
      \text{Uniform} [G^i(y_-), G^i(y)] \quad &\text{ if } G^i(y)> G^i(y_-)
    \end{cases}
\end{equation*}
  Now, denote by $\tilde \Fvec, \Fvec'$, the joint distribution that we get from $\Gvec$ on replacing $G^i$ with  $\tilde F$ and $F'$, respectively.
  Then, the winning probabilities for the agents in these two cases are:
  \begin{align}
  p_i(\tilde \Fvec, \lambda) = &\int_0^{\bar x} \prod_{j\neq i} F^j(\tilde x + \lambda_i-\lambda_j) d\tilde F(\tilde x)\nonumber\\
  =&\int_0^{\bar x} \bbE_{G^{iu}(x)}\left[ \prod_{j\neq i} F^j({\tilde F}^{-1}(G^{iu}(x)) + \lambda_i-\lambda_j)\right] dG^i(x),\label{eq: abstract misreporting claim helper11}\\
  p_j(\tilde \Fvec, \lambda) = &\int_x \tilde F(x+\lambda_j-\lambda_i)\prod_{k\not\in\{i,j\}} F^k(x + \lambda_j-\lambda_k)dF^j(x), \quad\forall j\neq i\label{eq: abstract misreporting claim helper12}
  \end{align}
  \begin{align}
  p_i(\Fvec', \lambda) = &\int_x \bbE_{G^{iu}(x)}\left[\prod_{j\neq i} F^j(F'^{-1}(G^{iu}(x)) + \lambda_i-\lambda_j)\right]dG^i(x),\label{eq: abstract misreporting claim helper21}\\
  p_j(\Fvec', \lambda) = &\int_x F'(x+\lambda_j-\lambda_i)\prod_{k\not\in\{i,j\}} F^k(x + \lambda_j-\lambda_k)dF^j(x), \quad\forall j\neq i\label{eq: abstract misreporting claim helper22}
  \end{align}
  Since $\tilde F(x) \leq F'(x) \forall x\in[0,\bar x]$ by construction, $\tilde F^{-1}(p) \geq F'^{-1}(p) \forall p\in[0,1]$. It's easy to see that $\eqref{eq: abstract misreporting claim helper11} \geq \eqref{eq: abstract misreporting claim helper21}$, and $\eqref{eq: abstract misreporting claim helper12}\leq
  \eqref{eq: abstract misreporting claim helper22}$. Using this we have
  \begin{align*}
      \sum_{j\neq i} (p_j(\Gvec, \lambda)-p_j(\tilde\Fvec, \lambda))^+ \geq
      &\sum_{j\neq i} (p_j(\Gvec, \lambda)-p_j(\Fvec', \lambda))^+.
  \end{align*}

and \newshipra{substituting $F'$ by $G^i$ and again using $\eqref{eq: abstract misreporting claim helper11} \geq \eqref{eq: abstract misreporting claim helper21}$, and $\eqref{eq: abstract misreporting claim helper12}\leq
  \eqref{eq: abstract misreporting claim helper22}$}, we get 
  \begin{align*}
      p_j(\Gvec, \lambda) - p_j(\tilde \Fvec, \lambda) &\geq 0 \quad\forall j\neq i,\\
      p_i(\Gvec, \lambda) - p_i(\tilde \Fvec, \lambda) &\leq 0
  \end{align*}
  This shows that $\tilde F$ is the maximizer of the RHS of \eqref{eq: claim
  misreporting size gain helper 1} \newshipra{among all distributions that are at most $\Delta$ away from $G^i$}. Using this we have
  \begin{align*}
    &\max_{F^i: p_i(\Fvec, \lambda) - p_i(\Gvec, \lambda)\geq 0} \sum_{j\neq i} (p_j(\Gvec,
    \lambda)-p_j(\Fvec, \lambda))^+\\
      =&\sum_{j\neq i} (p_j(\Gvec, \lambda)-p_j(\tilde\Fvec, \lambda))^+ \\
      =&p_i(\tilde \Fvec, \lambda) - p_i(\Gvec, \lambda)\\
      =&\int_0^{\bar x} \prod_{j\neq i} F^j\left(x + \lambda_i-\lambda_j\right)d\tilde F(x)- \int_0^{\bar x}\prod_{j\neq i} F^j(x+\lambda_i-\lambda_j)dG^i(x)\\
      =&\int_{x_\Delta}^{\bar x} \prod_{j\neq i} F^j\left(x + \lambda_i-\lambda_j\right)dG^i(x) + \prod_{j\neq i} F^j\left(\bar x + \lambda_i-\lambda_j\right)\Delta- \int_0^{\bar x}\prod_{j\neq i} F^j(x+\lambda_i-\lambda_j)dG^i(x)\\
      \leq&\Delta
  \end{align*}

  \paragraph{Bounding the RHS of \eqref{eq: claim misreporting size gain helper 2}} 
  Let $\hat F(x)\coloneqq  \min(G^i(x)+\Delta, 1)$.
  Then we can use the same steps as above for LHS to
  show that $\hat F(x)$ maximizes $\sum_{j\neq i} (p_j(\hat\Fvec, \lambda) -
  p_j(\Gvec, \lambda))^+$, and that 
  $$p_j(\Gvec, \lambda) - p_j(\hat\Fvec, \lambda)\leq
  0\,\forall j\neq i$$
  $$p_i(\Gvec, \lambda) - p_i(\hat \Fvec, \lambda)\geq 0$$
  This shows that $\hat F$ is the maximizer of the RHS of \eqref{eq: claim
  misreporting size gain helper 2}.From there, we have 
  \begin{align*}
    &\max_{F^i: p_i(\Fvec, \lambda) - p_i(\Gvec, \lambda)< 0} \sum_{j\neq i} (p_j(\Fvec, \lambda) - p_j(\Gvec,
    \lambda))^+\\
      =&\sum_j (p_j(\hat\Fvec, \lambda) - p_j(\Gvec, \lambda))^+ \\
      =&p_i(\Gvec, \lambda) - p_i(\hat \Fvec, \lambda) \\
      =&\int_0^{\bar x}\prod_{j\neq i} F^j(x+\lambda_i-\lambda_j)dG^i(x)- \int_0^{\bar x} \prod_{j\neq i} F^j\left(x + \lambda_i-\lambda_j\right)d\hat F(x)\\
      =&\int_0^{\bar x}\prod_{j\neq i} F^j(x+\lambda_i-\lambda_j)dG^i(x)- \Delta\prod_{j\neq i} F^j\left(\lambda_i-\lambda_j\right) -\int_0^{x^\Delta} \prod_{j\neq i} F^j\left(x + \lambda_i-\lambda_j\right)dF(x)\\
      =&\int_0^{\bar x}\prod_{j\neq i} F^j(x+\lambda_i-\lambda_j)dG^i(x)- \int_0^{x^\Delta} \prod_{j\neq i} F^j\left(x + \lambda_i-\lambda_j\right)dF(x)\\
      \leq&\Delta
  \end{align*}
  where $x^\Delta = F^{-1}(1-\Delta)$.
\end{proof}
Using Claim~\ref{claim: misreporting size gain abstract 0} we can now easily
prove the original Claim~\ref{claim: misreporting size gain abstract}.
\begin{proof}[Proof of Claim~\ref{claim: misreporting size gain abstract}]
  First we construct the following sequence of distributions where for any
  two adjacent distributions they only differ on one coordinate.
  \begin{align*}
  \Gvec_0 & = \Gvec = G^1\otimes\ldots\otimes G^n, \\
  \Gvec_1 & = F^1\otimes G^2\otimes\ldots\otimes G^n, \\
  \Gvec_2 & =  F^1\otimes F^2\otimes G^3\otimes\ldots \otimes G^n,\\
  \ldots&\\
   \Gvec_{n} & = \Fvec.
  \end{align*}
  Then we can decompose the difference between $p(\Fvec, \lambda)$ and $p^*$
  into a sum of differences:
  \begin{align*}
       ||p(\Fvec, \lambda) - p(\Gvec, \lambda)||_1 &=  ||p(\Gvec_{n}, \lambda) - p(\Gvec_0, \lambda)||_1 \\
      & \leq \sum_{i=1}^{n}||p(\Gvec_{i}, \lambda)- p(\Gvec_{i-1}, \lambda)||_1\\
      & \leq 2n\Delta 
  \end{align*}
  where the last step follows from Claim~\ref{claim: misreporting size gain abstract 0}.
  Since $\sum_j (p_j(\Fvec,
  \lambda)-p_j(\Gvec, \lambda))^+ = \sum_j (p_j(\Gvec, \lambda)-p_j(\Fvec,
  \lambda))^+ = \frac{1}{2}||p(\Fvec, \lambda) - p(\Gvec, \lambda)||_1$, we have
  $$
  \sum_j (p_j(\Fvec,
  \lambda)-p_j(\Gvec, \lambda))^+ = \sum_j (p_j(\Gvec, \lambda)-p_j(\Fvec,
  \lambda))^+ \leq n\Delta
  $$
\end{proof}

\subsubsection{Proof of Claim~\ref{claim: misreporting allocation difference}}
\begin{proof}
  WLOG, assume that $\lambda_1-\lambda'_1 \geq \lambda_2-\lambda'_2 \ldots\geq
  \lambda_n-\lambda'_n$.  Let $Q_{jk} = \Omega_j\cap \Omega'_k $ be the ``error
  flow'' of items from $j$ to $k$. It is easy to see that for $k<j$,
  $\Xvec_j+\lambda_j > \Xvec_k+\lambda_k \implies \Xvec_j+\lambda'_j > \Xvec_k+\lambda'_k$. Therefore 
  $Q_{jk} = \emptyset $ for $k<j$. Then it follows that
  $$\Omega'_j\setminus \Omega_j \subseteq \bigcup\limits_{i: i<j}Q_{ij}\subseteq\bigcup\limits_{i:i<j}\bigcup\limits_{k:k\geq j} Q_{ik}.$$
  The right hand side above is the net outflow from the set $\{i: i<j\}$. However,
  we know that  each individual agents' net in flow is $p_j(\Fvec,
  \lambda')-p_j(\Fvec, \lambda)$, so we can bound the RHS by
  $$
  \bbP\left(\bigcup\limits_{i:i<j}\bigcup\limits_{k:k\geq j} Q_{ik}\right)\leq \sum_{i:i<j}p_j(\Fvec, \lambda')-p_j(\Fvec, \lambda)\leq \Delta
  $$
  \end{proof}


\subsection{Proof of Theorem~\ref{thm: individual regret}}
\begin{proof}
  Fix an epoch $k$, let $\Delta = \sqrt{\frac{1}{2n(L_k-1)} \log(\frac{2}{\delta})}$,
  $\hat\lambda = \lambda^*(\hat\Fvec_{L_k-1})$.
  \begin{align}
  \text{(Lemma~\ref{lem: dkw})} \qquad &\sup_x |\hat F_{L_k-1}(x)-F(x)| \leq \Delta& \text{w.p. } 1-\delta/2\nonumber\\ 
  \text{(Lemma~\ref{lem: utility bound from learning error})}\implies&\bbE[u_i(\Xvec, \Xvec, \lambda^*)] - \bbE[u_i(\Xvec, \Xvec, \hat\lambda)] \leq n\Delta \bar x\quad\forall i &\text{w.p. } 1-\delta/2\label{eq: individual regret helper 1}\\ 
  (\text{Chernoff bound})\implies&\bbE[u_i(\Xvec, \Xvec, \lambda^*)](L_{k+1}-L_k) - \sum_{t=L_{k}}^{L_{k+1}-1}u_i(\Xvec_t, \Xvec_t, \hat\lambda)]&\nonumber\\ 
  \leq &n\Delta \bar x(L_{k+1}-L_k) + \bar x\sqrt{\frac{(L_{k+1}-L_k)}{2}\log(\frac{2}{\delta})} &\text{w.p. } 1-\delta\nonumber\\
  =&\sqrt{2^k n \log(\frac{2}{\delta})}\bar x +\sqrt{2^{k-1}\log(\frac{2}{\delta})}\bar x&\text{w.p. } 1-\delta\label{eq: individual regret helper 2}
  \end{align}
  The above bounds the regret in one epoch if the algorithm
  does not terminate before the epoch ends. 
  It remains to show that the algorithm with high probability does not terminate too early. This involves showing that with high probability, 
  no agent hits their capacity constraint $p^*_jT$ significantly earlier than $T$, and that 
  the detection algorithm does not falsely trigger. 

  Continuing from \eqref{eq: individual regret helper 2}, \newshipra{for any time step $T'\le T$, we have}
  
  \begin{align}
      &T\bbE\left[u_i(\Xvec, \Xvec, \lambda^*)\right] - \sum_{t=1}^{T'} u_i(\Xvec_t, \Xvec_t, \lambda_{k_t}) &\nonumber\\
      \leq & \sum_{k=0}^{\log_2 T'} \left[ (L_{k+1}-L_k)\bbE[u_i(\Xvec, \Xvec, \lambda^*)] - \sum_{t=L_{k}}^{L_{k+1}-1}u_i(\Xvec_t, \Xvec_t, \lambda_{k})] \right]\nonumber\\
      & + (T-T')\bbE\left[u_i(\Xvec, \Xvec, \lambda^*)\right]&\nonumber\\
      = & \sum_{k=1}^{\log_2 T} \sqrt{n2^{k} \log(\frac{2}{\delta})} \bar x + \sqrt{\frac{ 2^k}{2}\log(\frac{2}{\delta})}\bar x + (T-T')\bar x&\text{w.p } 1-\delta\log_2T\nonumber\\
      \leq & 2\sqrt{n \log(\frac{2}{\delta})} \bar x\sum_{k=1}^{\log_2 T} \sqrt{2^k} + (T-T')\bar x&\text{w.p } 1-\delta\log_2T\nonumber\\
      \leq & \frac{2\sqrt{2}}{\sqrt{2}-1}\sqrt{n T \log(\frac{2}{\delta})} \bar x + (T-T')\bar x&\text{w.p } 1-\delta\log_2T\label{eq: individual regret helper 3}
  \end{align}
  where the second inequality follows from\eqref{eq: individual regret helper 2}
  and union bound. \newshipra{Now, since there are at most $\log_2(T)$ epochs for any $T'\le T$, above holds for all epochs and therefore for all $T'$ with probability $1-\delta \log_2(T)$.}  Now we show that with high probability, for all  $T' \leq T -
  \frac{2\sqrt{2}}{\sqrt[]{2}-1}\sqrt[]{nT\log(\frac{2}{\delta})}$ and for any fixed agent $i$, the constraint of total allocation to agent $i$ to be less than $p^*_i T$ will be satisfied.  Note that a byproduct of applying Lemma~\ref{lem: utility bound
  from learning error} in \eqref{eq: individual regret helper 1}
  is that $|p_i(\Fvec, \lambda_k) - p^*_i|\leq n\Delta_{L_k-1}$ (See \eqref{eq:
  utility bound from learning error helper}). Fix a time step $\tau$,
  \begin{align*}
    &\sum_{t=1}^{\tau} \1[\argmax_j \Xvec_j+\lambda_{k_t j} = i]&\\
    (\text{Chernoff})\quad \leq &\sum_{k=1}^{\log_2 \tau} (L_{k+1}-L_k)p_i(\Fvec, \lambda_k) + \sqrt{\frac{(L_{k+1}-L_k)}{2}\log(\frac{2}{\delta})}&\text{w.p. } 1-\delta \log_2\tau\\
    \leq &\sum_{k=1}^{\log_2 \tau} (L_{k+1}-L_k)(p^*_i + n\Delta_{L_k-1})+ \sqrt{\frac{(L_{k+1}-L_k)}{2}\log(\frac{2}{\delta})}&\text{w.p. } 1-\delta \log_2\tau\\
    (L_k=2^k)\quad\leq &p^*_i\tau + \sum_{k=1}^{\log_2 \tau}  \left(\sqrt{n2^k\log(\frac{2}{\delta})}+ \sqrt{2^{k-1} \log(\frac{2}{\delta})}\right)&\text{w.p. } 1-\delta \log_2\tau\\
    \leq &p^*_i\tau + \frac{2\sqrt{2}}{\sqrt[]{2}-1}\sqrt[]{n\tau\log(\frac{2}{\delta})} &\text{w.p. } 1-\delta \log_2\tau
  \end{align*}

  This means that for all $\tau \leq T -
  \frac{2\sqrt{2}}{\sqrt[]{2}-1}\sqrt[]{nT\log(\frac{2}{\delta})} $, with
  probability $1-\delta\log_2 T$, 
  $$\sum_{t=1}^{\tau} \1[\argmax_j \Xvec_j+\lambda_{k_t j} =
  i]\leq p^*_iT,$$

  Combining above with \eqref{eq: individual
  regret helper 3}, we have that with probability
  $1-2\delta\log_2T$, for any fixed $i$,  if the algorithm terminates at $T'$ due to allocation limit reached for agent $i$, then $T'\ge \frac{2\sqrt{2}}{\sqrt[]{2}-1}\sqrt[]{nT\log(\frac{2}{\delta})}$, so that 
  $$
  T\bbE\left[u_i(\Xvec, \Xvec, \lambda^*)\right] - \sum_{t=1}^{T'} u_i(\Xvec, \Xvec, \lambda_{k_t}) \leq 
  \frac{2\sqrt{2}}{\sqrt{2}-1}\sqrt{n T \log(\frac{2}{\delta})} \bar x + \frac{2\sqrt{2}}{\sqrt[]{2}-1}\sqrt[]{nT\log(\frac{2}{\delta})}\bar x
  $$

  Finally, we also have to bound the probability that the detection algorithm
  falsely triggers. For a given time $t$ and for each $i$, let 
  \begin{align*}
  F^i_t(x) = &\frac{1}{t}\sum_{t=1}^t \1[X_{i,t}\leq x]\\
  \tilde F_t(x) = &\frac{1}{t(n-1)}\sum_{t=1}^t \sum_{j\neq i}\1[\Xvec_{j,t}\leq x]
  \end{align*}
  be the empirical CDF for agent $i$ and the rest of the agents. Since all
  agents are truthful, using Lemma~\ref{lem: dkw} we have that with probability
  $1-\delta$,
  \begin{align*}
    \sup_x|F^i_t(x)-F(x)|\leq & \sqrt{\frac{1}{2t}\log(\frac{2}{\delta})}\\
    \sup_x|\tilde F_t(x)-F(x)| \leq & \sqrt{\frac{1}{2t(n-1)}\log(\frac{2}{\delta})}
  \end{align*}
  This means that
  $\sup_x|F^i_t(x)-\tilde F_t(x)|\leq
  \sqrt{\frac{1}{t}\log(\frac{2}{\delta})}\leq
  32\sqrt{\frac{1}{t}\log(\frac{256et}{\delta})} = \Delta_t/2$, which
  means that Algorithm~\ref{algo: detection} is not triggered by agent $i$.
  Using union bound, we know that with probability $1-\delta nT$, the algorithm
  will not end early because of a false trigger (by any agent).

  The result follows by replacing $\delta$ with $\frac{\delta}{n(2\log_2T+T)}$
  and take the union bound over all agents.

\end{proof}

\section{Proof of Theorem~\ref{thm: main BIC theorem}}
\label{sec: proof of main bic theorem}
\subsection{Proof of Lemma~\ref{lem: reported distributions have to be close to true distribution}}
\begin{proof}
  Let $\alpha=\frac{\Delta}{4}$. We first check that the given condition on
  $\Delta$ satisfies $ \left(\frac{128et}{\alpha}\right)e^{-t\alpha^2/128} \leq
  \frac{\delta}{2}$ and that $2e^{-2t(n-1)\alpha^2}\leq \frac{\delta}{2}$

  \begin{align*}
          &\left(\frac{128et}{\alpha}\right)e^{-t\alpha^2/128} \leq
      \frac{\delta}{2}\\
      \iff& \alpha^2\geq \frac{128 \log(\frac{256 et}{\delta})}{t} + \frac{64}{t}\log(\frac{1}{\alpha^2})\\
      \impliedby& \alpha^2\geq \frac{256 \log(\frac{256 et}{\delta})}{t} \\
      \iff& \Delta\geq 64\sqrt{\frac{ \log(\frac{256 et}{\delta})}{t}} \\
  \end{align*}
  \begin{align*}
      &2e^{-2t(n-1)\alpha^2}\leq \frac{\delta}{2} \\
      \iff &\alpha\geq \sqrt{\frac{1}{2t(n-1)}\log(\frac{4}{\delta})}\\
      \impliedby& \Delta\geq 64\sqrt{\frac{ \log(\frac{256 et}{\delta})}{t}} \\
  \end{align*}

  Let $\bar F_t(x) = \frac{1}{t} \sum_{s=1}^t \1[\tilde X_{i,s} \leq x]$ be the
  empirical CDF of the samples collected from agent $i$. Let 
  $\tilde F_t(x) = \frac{1}{(n-1)t} \sum_{s=1}^t\sum_{j\neq i} \1[\tilde X_{j,s} \leq x]$ be the
  empirical CDF of all reported values from the other agents.  
  Let $\bar F(x) = \frac{1}{t}\sum_{s=1}^tF_s(x)$, where $F_s(x) = \bbP( \tilde X_{i, s} \leq x | \cH_s)$.
  Lemma~\ref{lem: martingale uniform convergence} tells us that with probability $1-\delta/2$,
  \begin{equation}
    \label{eq: lem reported distributions have to be close helper}
    \sup_x |\bar F_t(x) - \bar F(x) |\leq \frac{\Delta}{4}
  \end{equation}
  Since other agents are truthful, their reported values are independent, and we can use the regular DKW inequality to bound the empirical distribution constructed from their values. 
  Using Lemma~\ref{lem: dkw} we can show that with probability $1-\delta/2$,
  $$
  \sup_x |\tilde F_t(x) - F(x) |\leq \frac{\Delta}{4}.
  $$
  Using union bound, we can conclude that \newshipra{if $\sup_x |\bar F(x) - F(x)| \ge \Delta$, then} with probability $1-\delta$:
  $$\sup_x |\tilde F_t(x) - \bar F_t(x) |> \frac{\Delta}{2}$$
  which means that Algorithm~\ref{algo: detection} would have returned Reject.
\end{proof}

\subsection{Proof of Lemma~\ref{lem: utility gain bound given fixed allocation policy}}
\label{sec: proof of lem utility gain bound given fixed allocation policy}
First we state a technical result on monotone mapping between two distributions.
Given a cumulative distribution function $F$, we define the following random function:
\begin{equation}
    F^u(y) = 
    \begin{cases}
      F(y) &\text{ if } F(y) = F(y_-)\\
      \text{Uniform} [F(y_-), F(y)] \quad &\text{ if } F(y)> F(y_-)
    \end{cases}
    \label{eq: smoothed F CDF}
\end{equation}
If $F$ is a continuous distribution then $F^u$ is deterministic and is the
same as $F$. However if $F$ contains point masses, then at points where $F$
jumps, $F^u$ is uniformly sampled from the interval of that jump. It is
easy to see that $F^u$ has the nice property that if $Y\sim F$, then
$F^u(Y)\sim \text{Uniform}[0,1]$. 
\begin{restatable}{claim}{monotonecoupling}
Let $G$ be any distribution (cdf) over $\cX\subseteq \bbR$, and $F$ over $\cY\subseteq
\bbR$. Then there exists a unique joint distribution $r$ over $\cX\times \cY$ with
marginals $G,F$ such that the conditional distribution $r(\cdot | Y)$ has the following
\emph{monotonicity} property: define $\upperx_r(\cdot), \lowerx_r(\cdot)$ so that $X \in [\lowerx_r(Y), \upperx_r(Y)]$ almost surely, i.e.,
\begin{align*}
  \upperx_r(y) = \inf\{x: \bbP(X>x | Y=y) = 0\}\\
  \lowerx_r(y) = \sup\{x: \bbP(X < x| Y=y) = 0\},
\end{align*}
then

$$\upperx_r(y_1) \leq \lowerx_r(y_2) \quad\forall y_1<y_2.$$
In particular, the random variable $X|Y \sim r(\cdot|Y)$ can be sampled as 
$G^{-1}(F^u(Y))$, where $F^u$ is the random function defined in \eqref{eq:
smoothed F CDF} and $G^{-1}\coloneqq \inf\{x\in \bbR : G(x)\geq p\}$ denotes the generalized inverse, sometimes also referred to as the quantile function. 
\label{claim: generalized monotone mapping between two distributions}
\end{restatable}
The proof of this Claim is in Appendix~\ref{sec: proof of claim generalized monotone mapping between two distributions}. Using the above result, we derive the following key result that will provide insight into a strategic agent's best response to a greedy allocation strategy. Note that given a particular marginal distribution $G$ for the agent $i$'s reported values and the true value distribution $F$, there are many potential joint distributions between the true  and reported valuations. 
In the following lemma, we show that the "best" joint distribution among these, in terms of agent $i$'s utility maximization, is the one characterized in Claim \ref{claim: generalized monotone mapping between two distributions}.
\begin{claim}
\label{claim: best reporting function}
Fix a greedy allocation policy $\lambda$. Let $\Xvec \in [0,\bar x]^n$ be drawn from $F\otimes\ldots\otimes F$. 
Fix another distribution $G$ over $[0,\bar x]$. 
Given $\Xvec$, define $\tilde \Xvec^*$ as follows: let $\tilde X_{i}^*=G^{-1}(F^u(X_{i}))$, and $\tilde X_{j}^* = X_{j} \,\forall j\neq i$.
Let $\cR$ be the set of all joint distributions over $[0,\bar x]^2$ such that the marginals are $F$ and $G$; \newshipra{and for any $r\in \cR$, given $\Xvec$ define $\tilde \Xvec^r$ as follows: $\tilde X^r_i \sim r(\cdot|X_i)$, and $\tilde X^r_{j}=X_{j} \,\forall j\neq i$.}
Then
$$\bbE[u_i(\tilde \Xvec^*, \Xvec, \lambda)] \geq \max_{r\in \cR} \bbE[u_i(\tilde \Xvec^r, \Xvec, \lambda)].$$
\end{claim}
\begin{proof}
First we show that for any joint distribution that is not monotone \newshipra{(i.e., does not have the monotonicity property defined in Claim \ref{claim: generalized monotone mapping between two distributions})}, there is a monotone one that obtains at least as much utility.
Suppose $r$ is one such joint distribution that is not monotone, i.e.,
$\exists x_1<x_2, $ s.t. $ \upperx_r(x_1) > \lowerx_r(x_2)$ (as defined in Claim~\ref{claim: generalized monotone mapping between two distributions}). First recall that since  
\newshipra{$X_j\sim F, \forall j$ are independent,} the expected utility can be written as the following:
\begin{align*}
    \bbE[u_i(\tilde \Xvec^r, \Xvec, \lambda)]
    =\int_0^{\bar x} \int_{\lowerx_r(x)}^{\upperx_r(x)} x\prod_{j\neq i} F(\tilde x + \lambda_i-\lambda_j) dr(\tilde x |x) dF(x)
\end{align*}

Now consider a pair of values $\tilde x_1 >\tilde x_2$  such that $(\tilde x_1, x_1)$ and $(\tilde x_2, x_2)$  \newshipra{ has a non-zero probability density under distribution $r$} 
 . This pair exists because $\upperx_r(x_1) > \lowerx_r(x_2)$. Then using the fact that for
$a,b,c,d>0, a<b, c<d: ac+bd > ad+bc$, we can see that:
$$
x_1\prod F(\tilde x_1 + \lambda_i-\lambda_j) + x_2\prod F(\tilde x_2 + \lambda_i-\lambda_j) <
x_1\prod F(\tilde x_2 + \lambda_i-\lambda_j) + x_2\prod F(\tilde x_1 + \lambda_i-\lambda_j)
$$
This means that if we exchanged the probability mass between the two conditionals of $x_1, x_2$, the utility would be at least as much as
before, if not higher. This means that at least one monotone joint distribution belongs in the 
set of utility maximizing joint distributions. 
Since Claim~\ref{claim: generalized monotone mapping between two distributions} showed that the distribution of ($G^{-1}(F^u(X)), X)$ is the unique joint distribution that is monotone, we conclude that $\tilde \Xvec^*$  as defined in the lemma statement is indeed utility maximizing.  
\end{proof}
\paragraph*{Proof of Lemma~\ref{lem: utility gain bound given fixed allocation policy}}
\begin{proof}
  Let $G(x) \coloneqq (F(x)-\Delta)^+ \forall x < \bar x$,
  $ G(\bar x) \coloneqq 1$ be the distribution whose CDF is shifted down
  from $F$ by $\Delta$. Let $\tilde r$ be the utility maximizing joint
  distribution from Claim~\ref{claim: best reporting function}. Let $\hat r$, $\hat F$ be a different pair of joint and marginal distribution
  such that $\sup_x |F(x) - \hat F(x)| \leq \Delta$. We know
  that $\hat F(x)\geq G(x)$ for all $x$. Agent $i$'s utilities for using
  $\hat r$ and $\tilde r$ respectively, are:
\begin{align}
    \bbE_{\hat r}[u_
    i(\hat \Xvec, \Xvec, \lambda)] = &\int_0^{\bar x}x\int_{\lowerx_{\hat r}(x)}^{\upperx_{\hat r}(x)} \prod\limits_{j\neq i} F(\hat x + \lambda_i-\lambda_j) d\hat r(\hat x|x)dF(x)\nonumber\\
    =&\int_0^{\bar x} x\bbE_{F^u(x)}\left[\prod\limits_{j\neq i} F\left(\hat F^{-1}(F^u(x)) + \lambda_i-\lambda_j\right)\right] dF(x) \label{eq: alternative utility}
\end{align}
and 
\begin{align}
    \bbE_{\tilde r}[u_i(\tilde \Xvec, \Xvec, \lambda)] 
    =&\int x\bbE_{F^u(x)}\left[\prod\limits_{j\neq i} F\left( G^{-1}(F^u(x)) + \lambda_i-\lambda_j\right) \right]dF(x) \label{eq: delta below utility}
\end{align}
respectively.
Since $\hat F(x)\geq G(x)$, we know $\hat F^{-1}(p)\leq G^{-1}(p)$. Clearly $\eqref{eq: alternative utility}\leq \eqref{eq: delta below utility}$. \newshipra{We conclude that given a  greedy allocation policy $\lambda$, true valuation $X_{i,t}$ and truthful agents $j\ne i$ (with $\tilde X_{j,t} = X_{j,t}$),  reporting $\tilde X_{i,t} \sim \tilde r(\cdot|X_{i,t})$ is a strategy for agent $i$ that maximizes $\bbE[u_i(\tilde \Xvec_t, \Xvec_t, \lambda)]$ subject to the marginal distribution constraint $\sup_x |F(x) - F_r(x)| \leq \Delta$.
That is,
$$ \bbE_r[u_i(\tilde \Xvec_t, \Xvec_t, \lambda) ]  \le  \bbE_{\tilde r}[u_i(\tilde \Xvec_t, \Xvec_t, \lambda) ] \quad\forall r \text{ s.t. } \sup_x|F_r(x) - F(x)|\leq \Delta$$
}
It remains to bound the difference
$\bbE_{\tilde r}[u_i(\tilde \Xvec, \Xvec, \lambda)] - \bbE[u_i(\Xvec, \Xvec, \lambda)]$. First note that $G^{-1}(p) = F^{-1}(p+\Delta)$. Then we have that
\begin{align}
    &\bbE_r[u_i(\tilde \Xvec, \Xvec, \lambda)] - \bbE[u_i(\Xvec, \Xvec, \lambda)] \label{eq: profit from strategizing}\\
    =&\int_0^{\bar x} x\left(\bbE_{F^u(x)}\left[\prod_{j\neq i} F\left(F^{-1}(F^u(x)+\Delta) + \lambda_i-\lambda_j\right)\right] - \prod_{j\neq i} F(x+\lambda_i-\lambda_j)\right)dF(x)\nonumber\\
    \leq &\bar x\int_0^{\bar x} \left(\bbE_{F^u(x)}\left[\prod_{j\neq i} F\left(F^{-1}( F^u(x)+\Delta) + \lambda_i-\lambda_j\right)\right] - \prod_{j\neq i} F(x+\lambda_i-\lambda_j)\right)dF(x)\label{eq: profit from strategizing 2}
\end{align}
where the inequality follows from the fact that $F^{-1}(F^u(x)+\Delta)\geq x$ $w.p.1$ for
all $x$. To bound the remaining expression in the integral, we can use the fact that \newshipra{since the marginal distribution of $\tilde x$ under the joint distribution $\tilde r(\tilde x, x)$ is $G$, we have }
\begin{align}
    &\int_0^{\bar x}\int_0^{\bar x} \prod_{j\neq i} F\left(\tilde x + \lambda_i-\lambda_j\right) d\tilde r(\tilde x|x)dF(x)\nonumber\\
    =&\int_0^{\bar x} \prod_{j\neq i} F\left(x + \lambda_i-\lambda_j\right)dG(x)\nonumber\\
    =&\int_{x_\Delta}^{\bar x} \prod_{j\neq i} F\left(x + \lambda_i-\lambda_j\right)dF(x) + \prod_{j\neq i} F\left(\bar x + \lambda_i-\lambda_j\right)\Delta\nonumber\\
    \leq&\int_{x_\Delta}^{\bar x} \prod_{j\neq i} F\left(x + \lambda_i-\lambda_j\right)dF(x) + \Delta\label{eq: profit bound helper 1}
\end{align}
where $x_{\Delta} := F^{-1}(\Delta)$. Similarly,
\begin{align}
    &\int_0^{\bar x} \prod_{j\neq i} F\left(x + \lambda_i-\lambda_j\right) dF(x)\nonumber\\
    =&\int_{0}^{x_{\Delta}} \prod_{j\neq i} F\left(x + \lambda_i-\lambda_j\right)dF(x) + \int_{x_{\Delta}}^{\bar x} \prod_{j\neq i} F\left(x + \lambda_i-\lambda_j\right)dF(x)\nonumber\\
    \geq&\int_{x_\Delta}^{\bar x} \prod_{j\neq i} F\left(x + \lambda_i-\lambda_j\right)dF(x)\label{eq: profit bound helper 2}
\end{align}

Plugging \eqref{eq: profit bound helper 1} and \eqref{eq: profit bound helper 2}
back to \eqref{eq: profit from strategizing 2}, we can now bound the expression in  \eqref{eq: profit from strategizing}, and thereby the profit from strategizing, by $\bar x\Delta$.

\end{proof}
\subsection{Proof of Lemma~\ref{lem: learning under strategic report}}
\begin{proof}
  Let $\bar F$ be the average distribution that agent $i$
  reported from up to round $T'$: $\bar F = \frac{1}{T'} \sum_{t=1}^{T'}F_t$,
  where $F_t$ is the reported value distribution of agent $i$ in time $t$: $ F_t(x)\coloneqq \bbP(\tilde X_{i,t} \leq x | \cH_t)$.
  Since the the detection algorithm has not been triggered, 
  we can conclude using Lemma~\ref{lem: reported distributions have to be close
  to true distribution} that with probability $1-\delta$, 
  \begin{align*}
    &\sup_x|\bar F(x) - F(x)| < \Delta \coloneqq 64\sqrt{\frac{\log(\frac{256eT'}{\delta})}{T'}},\\
    \text{and}\quad &\sup_x|\bar F_{T'}(x) - \bar F(x)| < \frac{\Delta}{4} = 16\sqrt{\frac{\log(\frac{256eT'}{\delta})}{T'}}.
  \end{align*}
  The second inequality holds because the proof of Lemma~\ref{lem: reported
  distributions have to be close to true distribution} uses the second
  inequality to show the first (see Equation~\ref{eq: lem reported distributions have
  to be close helper}).
  Combining the above two steps, we have 
  \begin{equation}
    \label{eq: single strategic agent utility bound helper 1}
    \sup_x|\bar F_{T'}(x) - F(x)| < \frac{\Delta}{4} = 80\sqrt{\frac{\log(\frac{256eT'}{\delta})}{T'}} \qquad\text{w.p. }1-\delta.
  \end{equation}
  This shows that if the detection algorithm has not been triggered, the
  empirical CDF of strategic agent's reported values are close to the true CDF.
  Let $\tilde F_{T'}(x) = \frac{1}{(n-1)T'} \sum_{t=1}^{T'}\sum_{j\neq i} \1[X_{j,t} \leq x]$ be
  the emipircal distriution from all agents other than $i$. We know from
  Lemma~\ref{lem: dkw} that
  \begin{equation}
    \label{eq: single strategic agent utility bound helper 2}
    \sup_x |\tilde F_{T'}(x) - F(x)| \leq
    \sqrt{\frac{1}{2(n-1)T'}\log(\frac{2}{\delta})}\qquad \text{w.p. } 1-\delta.
  \end{equation}
  Combining \eqref{eq: single strategic agent utility bound helper 1} and
  \eqref{eq: single strategic agent utility bound helper 2}, we can now bound the
  error in the combined estimation,
  $\hat F_{T'} =  \frac{1}{nT'} \sum_{t=1}^{T'}\sum_{j=1}^n \1[\Xvec^t_j \leq x] $:
  \begin{align}
    &\sup_x |\hat F_{T'}(x) - F(x)|&\nonumber\\
    =&\sup_x |\frac{1}{n} \bar F_{T'}(x) + \frac{n-1}{n}\tilde F_{T'}(x) - F(x)|&\nonumber\\
    =&\sup_x |\frac{1}{n} \bar F_{T'}(x) - \frac{1}{n}F(x) + \frac{n-1}{n}\tilde F_{T'}(x) - \frac{n-1}{n}F(x)|&\nonumber\\
    \leq&\sup_x |\frac{1}{n} \bar F_{T'}(x) - \frac{1}{n}F(x)| + \sup_x|\frac{n-1}{n}\tilde F_{T'}(x) - \frac{n-1}{n}F(x)|&\nonumber\\
    \leq&80\sqrt{\frac{\log(\frac{256eT'}{\delta})}{nT'}} + \sqrt{\frac{1}{2nT'}\log(\frac{2}{\delta})} &\text{w.p. }1-2\delta\nonumber\\
    \leq&81\sqrt{\frac{\log(\frac{256eT'}{\delta})}{nT'}}&\text{w.p. }1-2\delta
    \label{eq: single strategic agent utility bound helper 3}
  \end{align}

  Let $\hat \Fvec_{T'} = \hat F_{T'}\otimes\ldots\otimes \hat F_{T'}$, and $\lambda =
  \lambda^*(\hat \Fvec_{T'})$, and $\Delta_{T'}= 81\sqrt{\frac{\log(\frac{256eT'}{\delta})}{nT'}}$. 
  Applying Lemma~\ref{lem: utility bound from learning error} to \eqref{eq: single strategic agent utility bound helper 3} we have
  \begin{align}
   &\sup_x |\hat F_{T'}(x)-F(x)| \leq \Delta_{T'}\quad \text{w.p. } 1-2\delta\nonumber\\ 
  (\text{Lemma~\ref{lem: utility bound from learning error}})\implies&\bbE[u_i(\Xvec, \Xvec, \lambda)]-\bbE[u_i(\Xvec, \Xvec, \lambda^*)] \leq n\Delta_{T'} \bar x\nonumber
  \end{align}
\end{proof}

\subsection{Proof of Lemma~\ref{lem: single strategic agent utility bound}}
\begin{proof}
  Let $F_t, t=1,\ldots, T$ be the distributions that agent $i$ reports from in
  each round given the history, i.e. $\tilde X_{i,t} | \cH_t\sim F_t$.  First we try to bound the
  utility that the strategic agent can get from a single epoch.  Fix an epoch
  $k$. Suppose $T'$ is the time when either detection algorithm is triggered, or
  the first time some receiver hits his allocation budget of $p^*_jT$. Let $\tau
  = \min(T', L_{k+1}-1)$. We now define three distributions:
  \begin{align*}
  \bar F^1 =& \frac{1}{L_k-1}\sum_{t=1}^{L_k-1} F_t\\
  \bar F^2 =& \frac{1}{\tau-1}\sum_{t=1}^{\tau-1} F_t\\
  \bar F^3 =& \frac{1}{\tau-L_k}\sum_{t=L_k}^{\tau-1} F_t
  \end{align*}
  These are the average distributions that agent $i$ reported from, averaged
  across three time periods: $[1,L_k)$, $[1, \tau)$ and $[L_k,\tau)$. In particular, $\bar F^3$ is the average distribution that the
  strategic agent reports from in epoch $k$.
  From Lemma~\ref{lem: reported distributions have to be close to true distribution}
  we know that with probability $1-2\delta$:
  \begin{align*}
  \sup_x|\bar F^1(x) - F(x)| \leq&64\sqrt{\frac{\log(\frac{256e(L_k-1)}{\delta})}{n(L_k-1)}}\\
  \sup_x|\bar F^2(x) - F(x)| \leq&64\sqrt{\frac{\log(\frac{256e(\tau-1)}{\delta})}{n(\tau-1)}}
  \end{align*}
  which together means that 
  \begin{align*}
    &\sup_x|\bar F^2(x) - F(x)| = \sup_x|\frac{L_k}{\tau}(\bar F^1(x) - F(x)) + \frac{\tau-L_k}{\tau} (\bar F^3(x) - F(x))| \\
    \implies &\sup_x|\bar F^2(x) - F(x)| \geq  \sup_x|\frac{\tau-L_k}{\tau} (\bar F^3(x) - F(x))| - \sup_x|\frac{L_k}{\tau}(\bar F^1(x) - F(x))|\\
    \implies &\sup_x|\bar F^3(x) - F(x)| \leq  \bar\Delta_k \coloneqq \min\left(\frac{128\tau}{\tau-L_k}\sqrt{\frac{\log(\frac{256e(\tau-1)}{\delta})}{n(\tau-1)}}, 1\right)
  \end{align*}
  Note that the last step also uses the fact that the difference between two
  CDFs cannot be bigger than 1.  
  Let $r$ be any joint distribution for agent $i$'s reported  and true
  valuation $(\tilde x, x)$ such that the marginal for the reported valuation is equal to $\bar F^3$, i.e., 
  $$\bar X_{i,t}\sim r(\cdot|X_{i,t}), X_{i,t}\sim F\implies F_r(x)\coloneqq \bbP(\bar X_{i,t}\leq x)= \bar F^3$$ 
  Let $\bar \Xvec$ denote the reported value vector when $i$ is the only strategic agent and uses $r(\cdot|X_i)$ to pick his reported value: $ \bar X_j= X_j \,\forall j\neq
  i$, $\bar  X_i \sim r( \cdot | X_i)$.  Let $\Delta_{L_k-1} =
  81\sqrt{\frac{\log(\frac{256e(L_k-1)}{\delta})}{n(L_k-1)}}$. Using this, we
  have
  \begin{align}
  (\text{Lemma~\ref{lem: learning under strategic report}})\implies&\bbE[u_i(\Xvec, \Xvec, \lambda)]-\bbE[u_i(\Xvec, \Xvec, \lambda^*)] \leq n\Delta_{L_k-1} \bar x \nonumber\\
  (\text{Lemma~\ref{lem: utility gain bound given fixed allocation policy}})\implies&\bbE[u_i(\bar \Xvec, \Xvec, \lambda)]-\bbE[u_i(\Xvec, \Xvec, \lambda^*)] \leq n\Delta_{L_k-1} \bar x + \bar\Delta_{k}\bar x\nonumber\\ 
  (\text{Corollary~\ref{cor: martingale uniform convergence}})\implies&\sum_{t=L_{k}}^{\tau-1}u_i(\tilde \Xvec_t, \Xvec_t, \tilde \lambda_{k_t}) - (\tau-L_k)\bbE[u_i(\Xvec, \Xvec, \lambda^*)]\nonumber\\
                                 &\leq (n\Delta_{L_k-1} +\bar\Delta_k)\bar x (\tau-L_k) 
                                 + 16\sqrt{(\tau-L_k)\log(\frac{128e(\tau-L_k)}{\delta})}\bar x &\text{w.p. } 1-\delta\nonumber\\
  &\leq 81\sqrt{\frac{n(\tau-L_k)^2}{2(L_k-1)} \log(\frac{256eL_k}{\delta})} \bar x + 144\sqrt{2\tau\log(\frac{256e\tau}{\delta})} \bar x  &\text{w.p. } 1-\delta \label{eq: individual gain helper}
  \end{align}
  The above is a high probability bound on how much an agent can get in one epoch. We can now bound the strategic agent's utility over the full horizon.
  \begin{align*}
      &\sum_{t=1}^{T'} u_i(\tilde \Xvec_t, \Xvec_t, \tilde \lambda_{k_t}) - T'\bbE\left[u_i(\Xvec, \Xvec, \lambda^*)\right] \\
      \leq & \sum_{k=0}^{\log_2 T'-1} \left[ \sum_{t=L_{k}}^{L_{k+1}-1}u_i(\tilde \Xvec_t, \Xvec_t, \lambda) - (L_{k+1} - L_k)\bbE[u_i(\Xvec, \Xvec, \lambda^*)] \right]\\
      (\text{Using } \eqref{eq: individual gain helper})\leq & \bar x (L_1-1) + \sum_{k=0}^{\log_2 T'-1} \left(81\sqrt{\frac{n(L_{k+1}-L_k)^2}{2(L_k-1)} \log(\frac{256e(L_k-1)}{\delta})} \bar x \right.\\
      &+ \left.  144\sqrt{2L_{k+1}\log(\frac{256eL_{k+1}}{\delta})} \bar x   \right)&\text{w.p } 1-\delta\log_2T\\
      (L_k = 2^k)\leq & \bar x + \sum_{k=0}^{\log_2 T'-1} 285\sqrt{n2^{k} \log(\frac{256eT'}{\delta})} \bar x &\text{w.p } 1-\delta\log_2T\\
      \leq & \left(\frac{285\sqrt{2}}{\sqrt{2}-1}\sqrt{nT' \log(\frac{256e}{\delta})} + 1\right)\bar x &\text{w.p } 1-\delta\log_2T
  \end{align*}

  The result follows by replacing the original $\delta$ with
  $\frac{\delta}{\log_2T}$.

\end{proof}

\section{Auxiliary Proofs}
\subsection{Proof of Claim~\ref{claim: generalized monotone mapping between two distributions}}
\label{sec: proof of claim generalized monotone mapping between two distributions}
\monotonecoupling*
\begin{proof}
  We first prove existence by constructing a joint distribution with the desired marginals and monotonicity, 
  then we show uniqueness.
  \paragraph{Existence.} We will construct the joint distribution by defining the conditional distribution of $X$ given $Y=y$ for every $y$. 
  Note that if $F$ is a continuous distribution, then we can easily construct $r(\cdot|Y=y)$ using the inverse-CDF method:
  $$
  X|y = G^{-1}(F(y)) 
  $$
  where $G^{-1}\coloneqq \inf\{x\in \bbR : G(x)\geq p\}$ is the generalized inverse. This works because $F(Y)\sim$~Uniform[0,1]. 
  If $F$ contains point masses, then $F(Y)$ is no longer uniformly distributed, and the inverse-CDF method does not work. To resolve this, we construct a different random variable $F^u(y)$ for each value $y$.
  For a given sample $y$, If $F(y)\neq F(y_-)$, let $F^u(y) \sim \text{Uniform}[F(y_-), F(y)]$. Otherwise, let $F^u(y) = F(y)$. Now we let
  $$X|y= G^{-1}(F^u(y))$$
  To see that $X$ sampled using this process has the marginal distribution $G$, we just need to show that $F^u(Y)$ is uniformly distributed. For a given $p$, if $\exists y \,s.t.\, F(y) = p$, then $\bbP(F^u(Y)\leq p) = \bbP(F(Y)\leq p) = \bbP(Y\leq y) = p$. Otherwise that means $ \exists y \,s.t.\, p_1\coloneqq F(y_-)\leq p\text{ and }p_2\coloneqq F(y)>p$.
  \begin{align*}
        &\bbP(F^u(Y)\leq p) \\
      = &\bbP(Y< y) + \bbP(F^u(y) \leq p| Y = y)\bbP(Y = y)\\
      = & p_1 + \frac{p-p_1}{p_2-p_1}(p_2-p_1)\\
      = &p
  \end{align*}
  This construction also satisfies monotonicity, since if $y_1<y_2$, then $F^u(y_1)\leq F(y_1)$ w.p.1. and $F^u(y_2) \geq F(y_1)$ w.p.1. 
  \paragraph{Uniqueness}
  Now we show uniqueness. 
  For a given $(x, y)$ pair, 
  suppose $x< \upperx_r(y)$. Then from monotonicity we know  $\lowerx_r(y') \geq \upperx_r(y)>x$ for all $y'>y$, which implies that  
  $$ \bbP_r(X\leq x, Y\leq y)  = G(x).$$

  If $x \geq \upperx_r(y)$, then from monotonicity we know $\upperx_r(y') \leq \upperx_r(y)\leq x $  for all $y'<y$, which implies that 
  $$
  \bbP_r(X\leq x, Y\leq y) = F(y)
  $$
  
  Since $G$ and $F$ are fixed, we have shown that all joint distributions $r$ with monotonicity and the required marginals are the same. 
  \end{proof}
\end{document}